\documentclass{IEEEtran}
\usepackage{cite}
\usepackage{amsmath,amssymb,amsfonts,amsthm}
\usepackage{graphicx}
\usepackage{textcomp}
\usepackage{xcolor}
\usepackage{enumerate}
\usepackage{array}
\usepackage{tabulary}

\newcommand{\col}{\operatorname{col}}
\newcommand{\diag}{\operatorname{diag}}

\newtheorem{assumption}{Assumption}
\newtheorem{lemma}{Lemma}
\newtheorem{theorem}{Theorem}
\newtheorem{definition}{Definition}
\newtheorem{remark}{Remark}
\newtheorem{corollary}{Corollary}

\def\BibTeX{{\rm B\kern-.05em{\sc i\kern-.025em b}\kern-.08em
		T\kern-.1667em\lower.7ex\hbox{E}\kern-.125emX}}
\def\mathbi#1{\textbf{\em #1}}

\begin{document}
	\title{Distributed Kalman Filter with Ultimately Accurate Fused Measurement Covariance}
	\author{Tuo Yang, Jiachen Qian, Zhisheng Duan*, \textit{Senior Member, IEEE}, Zhiyong Sun, \textit{Member, IEEE}
		\thanks{This work was supported by the National Natural Science Foundation of China under grants T2121002, U24A20266, 62173006. ($*Corresponding\;author: Zhisheng\;Duan)$.}
		
		\thanks{Tuo Yang, Jiachen Qian, Zhisheng Duan and Zhiyong Sun are with State Key Laboratory for Turbulence and Complex Systems, Department of Mechanics and Engineering Science, College of Engineering, Peking University, Beijing 100871, China. E-mails: duanzs@pku.edu.cn (Z. Duan), yangtuo@pku.edu.cn (T. Yang), jcq@pku.edu.cn (J. Qian), zhiyong.sun@pku.edu.cn (Z. Sun). } 
		}
	
	\maketitle
	\begin{abstract}
		This paper investigates the distributed Kalman filter (DKF) for linear systems, with specific attention on measurement fusion, which is a typical way of information sharing and is vital for enhancing stability and improving estimation accuracy. We show that it is the mismatch between the fused measurement and the fused covariance that leads to performance degradation or inconsistency in previous consensus-based DKF algorithms. To address this issue, we introduce two fully distributed approaches for calculating the exact covariance of the fused measurements, building upon which the modified DKF algorithms are proposed. 
		Moreover, the performance analysis of the modified algorithms is also provided under rather mild conditions, including the steady-state value of the estimation error covariance. We also show that due to the guaranteed consistency in the modified DKF algorithms, the steady-state estimation accuracy is significantly improved compared to classical DKF algorithms.
	 Numerical experiments are carried out to validate the theoretical analysis and show the advantages of the proposed methods.
	\end{abstract}
	\begin{IEEEkeywords}
		 Distributed estimation,  Consensus filters, Information fusion, Random approximation
	\end{IEEEkeywords}

	\section{Introduction}
	
	Due to notable advancements in low-cost wireless sensors in recent years, the deployment of wireless sensor networks (WSNs) has become highly cost-effective. This brings immense potential into applications of WSNs, such as cooperative target tracking\cite{jia2016cooperative} and large-scale environmental monitoring\cite{xie2012fully}. Consequently, there is a growing concern regarding the state estimation problem within such networks.
	
	To provide a fused estimate,  a straightforward architecture is the centralized way, where the filtering process requires an information fusion center to collect all the measurements of the whole sensor network. Through stacking all the measurements  into one augmented term and performing the classic Kalman filter, the optimal estimate of the target system state can be obtained. This framework is called centralized Kalman filter (CKF){\cite{anderson1979optimal}}. However, the centralized framework relies heavily on the capability of the fusion center, which may suffer from channel congestion and node failure, especially when the number of nodes in the network is large. In contrast, the distributed filtering framework, where each sensor node is designed as a local fusion center to process local measurement and neighbouring information,  shows superior flexibility and robustness to the centralized framework. The design of a distributed state estimator primarily consists of two aspects, local signal processing and neighbouring information fusion. The local signal processing applies classical estimating techniques, such as moving-horizon estimation (MHE)\cite{farina2010distributed,battistelli2018distributed}, Kalman  filter\cite{olfati2005distributed,olfati2007distributed,olfati2009kalman,ryu2023consensus,battistelli2014consensus,battistelli2014kullback,battistelli2016stability,battilotti2021stability,sun2004multi,zhang2018sequential,kar2010gossip,qin2020randomized,cattivelli2010diffusion,hu2011diffusion,he2018consistent,he2019distributed,wang2017convergence,qian2021fully,qian2022consensus,qian2022harmonic,he2020distributed,kamgarpour2008convergence,kamal2011generalized,kamalTAC2013,thia2013distributed,wu2018distributed,battistelli2018distributedtriggered}, etc. The information fusion protocol considers what kind of information to exchange within the network, such as estimates\cite{olfati2009kalman}, measurements\cite{olfati2005distributed}, information vector\cite{battistelli2014kullback} and the covariance matrices. Besides, the methods of node interaction are also important to investigate, with 
 	typical ways including sequential fusion\cite{sun2004multi,zhang2018sequential}, gossip\cite{kar2010gossip,qin2020randomized}, diffusion\cite{cattivelli2010diffusion,hu2011diffusion}, and average consensus\cite{olfati2005distributed,olfati2007distributed,olfati2009kalman,ryu2023consensus,battistelli2014consensus,battistelli2014kullback,battistelli2016stability
 		,battilotti2021stability,he2018consistent,he2019distributed,wang2017convergence,qian2021fully,qian2022consensus,qian2022harmonic,duanpeihu2022distributed,dong2022consensus,he2020distributed,kamgarpour2008convergence,kamal2011generalized,kamalTAC2013,thia2013distributed,wu2018distributed,battistelli2018distributedtriggered}. Other topics on the distributed estimation involve correlated noise\cite{duanpeihu2022distributed}, unknown covariance\cite{dong2022consensus}, state constraints\cite{he2019distributed},  time-varying communication topology\cite{wang2017convergence}, event-triggered mechanism\cite{battistelli2018distributedtriggered,qian2021fully}, distributed EKF\cite{battistelli2016stability}, quantized communications\cite{he2020distributed}, etc.
	
	In this paper, we mainly focus on the topic of consensus-based distributed Kalman filter, which has recently generated great research interest due to the excellent performance of the Kalman filter in linear system state estimation and the effectiveness of the average consensus technique in gathering global information in a fully distributed way.  
	
	Early work in this field can be tracked back to consensus-on-measurement algorithms (CM)\cite{olfati2005distributed, olfati2007distributed}, where average consensus technique on  local normalized measurements and  covariance is performed in order to approximate the CKF. However, CM suffers the disadvantage that it requires multiple communication steps between successive sampling instants to guarantee the stability of the filter. 	
	In \cite{olfati2007distributed}, the Kalman consensus filter (KCF) is introduced, focusing on reaching consensus on the estimates. It is shown in \cite{kamgarpour2008convergence} that KCF converges to the CKF when consensus is achieved. The generalized Kalman consensus filter (GKCF) \cite{kamal2011generalized} further refines the fusion by weighting each estimate according to its covariance. In the information consensus filter (ICF) \cite{kamalTAC2013}, the fusion involves the information vector which combines the prior estimate and the current measurement.
	The consensus-on-information method (CI) is proposed in \cite{battistelli2014kullback} by performing consensus on the local prior probability density functions, which shows the equivalency to traditional covariance intersection technique\cite{julier1997non} when the scenario is reduced to linear Gaussian system. The difference between ICF and CI lies in the design of the consensus protocol. The CI methods are proved to be stable with only single communication step in each iteration, and the consistency (the actual error covariance can be upper bounded by the locally estimated error covariance) of the estimation is also guaranteed. However, the process of information fusion induces correlations among the estimates from various sensor nodes, making it challenging to evaluate and optimize its performance in a fully distributed way.  
	
	With the above discussion, it is meaningful  to facilitate the application of distributed filtering technique with the  development of essential theories for analyzing and optimizing the performance. Regarding this point, Qian et al.\cite{qian2022consensus} provide a steady-state performance analysis of CM and demonstrate the exponential convergence of the performance gap between CM and CKF with the increase of fusion step in each iteration, by comparing the solution of discrete algebraic Riccati equation (DARE) with different measurement covariance matrices. Besides, Qian et al. \cite{qian2022harmonic} formulate the harmonic coupled Riccati equation (HCRE) from the CI algorithm and presented the closed-form steady-state error covariance for CI through solving the corresponding HCRE. From the perspective of performance optimization,  Battistelli et al.\cite{battistelli2014consensus} introduce a gain into the fused measurement in CI to compensate for the conservatism  induced by matrix intersection technique, and this variant is referred to as HCMCI (Hybrid Consensus on Measurements—Consensus on Information). He et al. \cite{he2018consistent} design adaptive consensus weights to optimize the bound of error covariance. Noack et al. \cite{noack2017decentralized} propose the inverse covariance intersection method to reduce the conservatism induced by covariance intersection. Battilotti el al. \cite{battilotti2021stability} deduce two distributed filters using the steady-state Kalman gain, which possess one-step stability and optimality guarantee with sufficient fusion. Thia et al.\cite{thia2013distributed} and Wu et al. \cite{wu2018distributed} exploit finite-time/minimum-time consensus technique to fuse the measurements, recovering the estimate of CKF when the number of consensus steps per iteration exceeds a certain threshold.

	Nonetheless, previous works \cite{qian2022consensus,qian2022harmonic,battistelli2014consensus,he2018consistent,noack2017decentralized,battilotti2021stability} mainly focus on directly analyzing the filtering performance and compensating the conservatism of the algorithm, while overlooking the structural defects that lead to the occurrence of conservatism and inconsistency. Thus the existing modifications, such as adding new weight\cite{battistelli2014consensus} or using steady-state gains\cite{battilotti2021stability}, can only slightly reduce the performance degradation to a limited extent, while the information fusion step remains to be non-optimal. What’s worse, such modifications need an offline computation phase to obtain additional global information compared with the original methods, which is undesirable as well. In general, we expect the proposed filter to exhibit the following properties: 1) it requires only the collective observability, 2) its performance converges to CKF as the number of fusion steps between successive sampling instants increases, 3) limited reliance on global information, and 4) each node can evaluate the accuracy of its own estimate.

	Motivated by the above discussions, this paper aims to investigate the fundamental reasons for conservatism and inconsistency in consensus-based DKFs, and aim to propose a novel mechanism to overcome these drawbacks. The contributions of this paper can be summarized as follows:

	 1) We  identify that the cause of performance degradation in the CM\cite{olfati2005distributed} and CI\cite{battistelli2014kullback} stems from the utilization of  approximated covariance in measurement fusion. To address this, we formulate a closed-form expression of the actual fused measurement covariance, which is in a quadratic weighted sum (QWS) form. Using the average consensus technique, we propose two novel methods for calculating QWS, and derive explicit convergence rate for both methods. These approaches can be readily applied to other scenarios where distributed computation of QWS is needed.
	 
	 2)The Modified CM/CI algorithms are also derived, based on an online calculation of the accurate fused measurement covariance. Compared with the original methods, the modified ones demonstrate superior performance in steady-state estimation accuracy and offer additional benefits such as smaller and consistent steady-state covariance bound, asymptotic optimality with respect to number of consensus steps per iteration, no need of network size, etc. These methods provide a paradigm for addressing flawed measurement fusion, offering valuable insights and useful tools for the structural optimization of DKF.
	
	 3) To guarantee the stability of the modified algorithms, we prove the convergence of the discrete algebraic Riccati recursion (DARR) and harmonic-coupled Riccati recursion (HCRR) where the parameter $R$ (measurement covariance matrix) is asymptotically convergent, under only mild conditions of connected graph and collective observability. Relevant analysis not only enriches the DARE/HCRE theory, but also provides a pattern when dealing with similar problems such as transferring a steady-state DKF from offline  parameter computation to online setup.
	
	The remainder of this paper is organized as follows. Section~II analyzes the performance degradation problem in the CM, CI and HCMCI, which is induced by mismatched fused covariance.
	Section~III proposes the direct and stochastic methods for computing the QWS and discusses the convergence rate. Section~IV proposes the Modified CM and Modified CI algorithms, along with stability and performance analysis. Section~V provides  numerical experiments that validate the theoretical analysis of the proposed methods and provide comparisons with other algorithms. Section~VI presents the conclusions and discusses potential directions for future work.
	
	Notation: The set of positive integers is denoted as \(\mathbb{Z}^+\), and  \(\mathbb{Z}_0^+ = \mathbb{Z}^+\cup\{0\} \). For $k\in \mathbb{Z}^+$, we use $\underline{k}$ to denote the set $\{1,2,\dots, k\}$.
	A positive semidefinite  matrix \(P\in \mathbb{R}^{n\times n}\) is denoted as \( P\ge 0\), and $P\ge Q$ means $P-Q\ge 0$. Correspondingly, \(P>0\) and $P > Q$ represent the strictly positive definite cases. Let $X_i \in \mathbb{R}^{m\times n} $ be arbitrary matrices, then
	$\col_{i=1}^N(X_i) \in \mathbb{R}^{Nm\times n}$ represents the matrix formed by stacking  all $X_i $ into a single column, {$\diag_{i=1}^N(X_i)\in \mathbb{R}^{Nm \times Nn}$} denotes the matrix created by placing $X_i$ along the diagonal, and $X_i^\dagger$ represents the Moore-Penrose inverse. 
	 $\mathbb{E}_{\{w_i\}}[x_i]$ stands for the expectation of $x_i$ with respect to $w_i$. We use $x \sim \mathbf{N}(\mu, P_0)$ to denote that $x$ follows a multivariate normal distribution with mean $\mu$ and covariance $P_0>0$. Let $P_1=WI_rW^T\ge 0$ and $\text{rank}(P_1)=r$, then $x\sim \mathbf{N}_d(\mu, P_1)$ denotes the degenerate multivariate normal distribution, i.e.,  $W^\dagger x\sim \mathbf{N}(W^\dagger\mu, I_r)$. $\delta_{ik}$ is the Kronecker delta, and $\otimes$ is the Kronecker product. $[x]_i$ denotes the $i$-th element of a vector $x$, and $[X]_{i,j}$ the $(i,j)$-th of a matrix $X$. For $\mathcal{L}\in\mathbb{R}^{n\times n}, l_{ij}^{(\gamma)}$ stands for the $(i,j)$-th element of $\mathcal{L}^\gamma$ ($\mathcal{L}$ to the power of $\gamma$).
	 
	\section{Problem Formulation}
	\subsection{Graph Theory}	
	The communication topology of the sensor network is defined by a triple \(\mathcal{G}= \left(\mathcal{N},\mathcal{E},\mathcal{L}\right)\), where \(\mathcal{N}=\{1,2,\dots,N\}\) is the vertices set,
	 \(\mathcal{E}\subseteq \mathcal{N}\times\mathcal{N}\) is the edge set, and \(\mathcal{L}=\left[l_{ij}\right]_{N\times N}\) is the adjacency matrix, assigning a weight to each edge in the network. Specifically, $l_{ii}>0$; \(l_{ij} > 0\) if \((i,j) \in \mathcal{E}\), in which case we call \(j\) an in-neighbour of \(i\) and denote it as $j \in \mathcal{N}_i$. Otherwise, \(l_{ij}=0\).
	 
	A graph is called undirected if $(i, j)\in \mathcal{E}$ implies $(j, i)\in \mathcal{E}$. For undirected graphs, we can design the weight matrix to be doubly stochastic, i.e., \(\boldsymbol{1}_N^T\mathcal{L}=\boldsymbol{1}_N^T\) and \(\mathcal{L}\boldsymbol{1}_N=\boldsymbol{1}_N\), with all the eigenvalues of $\mathcal{L}$ being real and having absolute values no larger than $1$\cite{xiao2005scheme}. An undirected graph is connected if, for every pair of nodes in $\mathcal{N}$, there exists a path that connects them. This is equivalent to the graph having a spanning tree. For a connected graph, $1$ is a simple eigenvalue of $\mathcal{L}$\cite{qian2022consensus}.
	
	\subsection{System Model}
	In this paper, we consider the state estimation problem of discrete linear time-invariant (LTI) system, which can be modeled as
	\begin{equation*}
		x_{k+1} = Ax_k + w_k.
	\end{equation*}
	Here  \(x_k\in \mathbb{R}^{n}\) represents the system state, \(A \in \mathbb{R}^{n\times n}\)  is the state-transition matrix, and \(w_k \in \mathbb{R}^n\) denotes the process disturbance.
	
	The observation model is formulated as
	\begin{equation}
		\label{eq:observation_model}
		y_{i,k}= C_i x_k +v_{i,k},
	\end{equation}
	where \(y_{i,k} \in \mathbb{R}^{m_i}, C_i \in \mathbb{R}^{m_i\times n}, v_{i,k} \in \mathbb{R}^{m_i}\)  are the measurement output, observation matrix and measurement noise of the $i$-th sensor node, respectively.
	
	The following conditions are assumed to hold throughout this paper.
	
	\begin{assumption}
		\label{assumption:sec2 noise model}
		$\forall i, j \in \mathcal{N}$, $\forall h, k \in \mathbb{Z}_0^+$, $\forall s,t \in\mathbb{Z}^+$,
		\begin{enumerate}[1)]
			\item  $x_0 \sim \mathbf{N}(\hat{x}_0, P_0)$, $w_k\sim \mathbf{N}(\boldsymbol{0}, Q)$, $v_{i,t} \sim \mathbf{N}(\boldsymbol{0}, R_i)$,  with $ P_0, Q, R_i>0$.
			\item $x_0, w_{k}, v_{i,t}$ are mutually independent, $w_k$, $v_{i,t}$ are white process, i.e.,  \\
			\begin{equation*}
				\mathbb{E}\left[
				\begin{bmatrix}
					w_k\\ v_{i,t} \\x- \hat{x}_0
				\end{bmatrix}
				\begin{bmatrix}
					w_h& v_{j, s}
				\end{bmatrix}
				\right] = \begin{bmatrix}
					\delta_{kh} Q & \boldsymbol{O}\\
					\boldsymbol{O} & \delta_{ij}\delta_{ts} R_i\\
					\boldsymbol{O} & \boldsymbol{O}
				\end{bmatrix}.
			\end{equation*}
		\end{enumerate}
	\end{assumption}
	
	\begin{assumption}
		\label{assumption: topology}
		The communication topology is a connected undirected graph.
	\end{assumption}
	
	\begin{assumption}
		\label{assumption: collective observability}
		The system pair $(A,\col_{i=1}^N(C_i))$ is  observable.
	\end{assumption}
	
	These assumptions are mild and essential in the context of consensus-based DKFs \cite{battistelli2014kullback}.
	Generally speaking, Assumption~\ref{assumption:sec2 noise model} states the process disturbance and observation noise are independent, which is reasonable in most practical scenarios. This assumption facilitates the  calculation of the posterior covariance. Assumption~\ref{assumption: topology} is to ensure that each node is able to acquire the global information through sufficient communication.
	Assumption~\ref{assumption: collective observability} refers to the collective observability of the sensor network, which is more relaxed than requiring observability at each individual node. This assumption is closely related to the stability of DKFs.

	\subsection{Review of the CKF Algorithm}
	With an information fusion center that receives the observations from all sensors, the standard Kalman filter provides the optimal state estimate with minimum error covariance. This kind of filtering structure is called CKF and can be summarized in  the following two steps, where $\hat{x}_{k|k-1},\hat{x}_{k|k}$ denote the prior and posterior estimate,  and $P_{k|k-1}, P_{k|k}$ represent the corresponding error covariance matrices, respectively.
	
	Prediction:
	\begin{equation}
		\label{eq: sec2 CKF predict}
		\begin{aligned}
			P_{k|k-1} &= AP_{k-1|k-1}A^T +Q,\\
			\hat{x}_{k|k-1} &= A \hat{x}_{k-1|k-1}.
		\end{aligned}
	\end{equation}
	
	Measurement update:
	\begin{equation}
		\label{eq:sec2 CKF measure update}
		\begin{aligned}
			P_{k|k} &  \textstyle= (P_{k|k-1}^{-1}+\sum_{j=1}^N C_j^T R_j^{-1} C_j),\\
			\hat{x}_{k|k} & \textstyle= P_{k|k} (P_{k|k-1}^{-1}\hat{x}_{k|k-1} + \sum_{j=1}^N C_j^T R_j^{-1} y_{j,k}).
		\end{aligned}
	\end{equation}	

		\subsection{The Optimal Measurement Update with Accurate Fused Measurement Covariance}
	In this section, the fused measurement and its exact covariance is introduced, with which the optimal measurement update is derived.
	
	To extend the CKF to a fully distributed manner, a natural way is to consider employing the average consensus technique to calculate the global information terms $\sum_{j=1}^N C_j^T R_j^{-1} C_j$ and $\sum_{j=1}^N C_j^T R_j^{-1} y_{j,k}$. Specifically, each node communicates with its neighbors and performs information fusion using the following formula.

	Initialization:
	\begin{equation}
		\label{Sec2C_eq3y0}
		\begin{aligned}
		&\tilde{y}_{i,k}^{(0)} :=   C_i^T R_i^{-1} y_{i,k}, &\tilde{C}_{i}^{(0)} :=   C_i^T R_i^{-1} C_i.
		\end{aligned}
	\end{equation}
	
	Information fusion:
	\begin{equation}
		\label{Sec2C_eq4yfusion}
		\begin{aligned}
			& \textstyle\tilde{y}_{i,k}^{(m)} = \sum_{j=1}^{N}l_{ij}\tilde{y}_{j,k}^{(m-1)}, 
			& \textstyle \tilde{C}_{i}^{(m)} = \sum_{j=1}^{N}l_{ij}\tilde{C}_{j}^{(m-1)},
		\end{aligned}
	\end{equation}
	where $m = \underline{\gamma}$  indicates the number of fusion iterations the corresponding quantity has undergone, and $\gamma$ represents the maximum fusion steps available between two sampling instants. We call $\tilde{y}_{i,k}^{(m)}$ the fused measurement and $\tilde{C}_{i}^{(m)}$ the fused observation matrix, which can be verified that
	\begin{equation*}
		\begin{aligned}
			&\textstyle\tilde{y}_{i,k}^{(\gamma)} = \sum_{j=1}^{N}l_{ij}^{(\gamma)} C_j^T R_j^{-1}y_{j,k}, &\textstyle\tilde{C}_i^{(\gamma)} =  \sum_{j=1}^{N}l_{ij}^{(\gamma)}C_j^T R_j^{-1}C_j, 
		\end{aligned}	
	\end{equation*}
	where  $l_{ij}^{(\gamma)}$ is the $(i,j)$-th element of $\mathcal{L}^\gamma$.
	
	The average consensus technique guarantees that as the number of consensus steps per iteration, $\gamma$, tends to infinity, the weights $l_{ij}^{(\gamma)}$ tends to $\frac{1}{N}$(\cite[Lemma~1]{qian2022consensus}). In this way, each node obtains \(\frac{1}{N}\sum_{j=1}^N C_j^T R_j^{-1} C_j\) and \(\frac{1}{N}\sum_{j=1}^N C_j^T R_j^{-1} y_{j,k}\). With the knowledge of the network size N, each node can locally perform \eqref{eq: sec2 CKF predict} and \eqref{eq:sec2 CKF measure update} to obtain estimate numerically identical to those of the optimal CKF algorithm.
	
	However, since the maximum number of consensus steps per iteration is typically limited in practice,  the case where \(\gamma\) remains finite and constant is of greater  practical relevance. 
	Once \(\gamma\) rounds of fusion have been completed, the fused information at each node is given by \( \tilde{y}_{i,k}^{(\gamma)}\) and \( \tilde{C}_{i,k}^{(\gamma)}\).
	 To simplify notation, we omit the superscript \((\gamma)\) here, writing \(\tilde{y}_{i,k} := \tilde{y}_{i,k}^{(\gamma)}\) and \(\tilde{C}_{i} := \tilde{C}_{i,k}^{(\gamma)}\). Exploiting the observation model \eqref{eq:observation_model}, we have
	 \begin{equation}
	 	\label{eq:fused measurement model}
	 	\tilde{y}_{i,k} = \tilde{C}_i x + \tilde{v}_{i,k},
	 \end{equation}
	 where \(\tilde{v}_{i,k} := \sum_{j=1}^{N} l_{ij}^{(\gamma)} C_j^T R_j^{-1} v_{j,k}\) represents the fused measurement noise. 
	 Since the weighted sum of multiple independent Gaussian random variables is still a Gaussian random variable\cite{johnson2002applied}, we can conclude that $\tilde{v}_{i,k}\sim \mathbf{N}_d(0,\tilde{R}_i)$, where 
	 \begin{equation}
	 	\textstyle\label{eq:covariance of  fuesed measurement}
	 	\tilde{R}_i := \mathbb{E}[\tilde{v}_{i,k}\tilde{v}_{i,k}^T] = \sum_{j=1}^{N}(l_{ij}^{(\gamma)})^2C_j^T R_j^{-1}C_j 
	 \end{equation}
	  is the exact covariance matrix of the fused measurement noise.
	
	In light of \eqref{eq:fused measurement model}, $\tilde{y}_{i,k} $ can be regarded as the measurement output obtained from one single virtual sensor, of which the observation matrix is  $\tilde{C}_i$, and measurement noise is $\tilde{v}_{i,k}$. By Assumption~\ref{assumption:sec2 noise model},  $w_k$ and $\tilde{v}_{i,k}$  are independent random variables. Therefore, given the prior estimate $\hat{x}_{i,k}^-$ satisfying $x_{k}\sim \mathbf{N}(\hat{x}_{i,k}^-, P_{i,k}^-)$ and $\mathbb{E}[(x_k-\hat{x}_{i,k}^-)\tilde{v}_{i,k}^T] = \boldsymbol{O}$, the optimal posterior estimate $\hat{x}_{i,k}^+$ and covariance $P_{i,k}^+$, in the sense of minimum of variance\cite{anderson1979optimal}, is 
	\begin{equation}
		\label{eq:minimum variance estimator}
		\begin{aligned}
			P_{i,k}^+ &= \left((P_{i,k}^{- })^{-1}+\tilde{C}_i^T \tilde{R}_i^{\dagger}\tilde{C}_i\right)^{-1},\\
			\hat{x}_{i,k}^+ &= P_{i,k}^+ \left((P_{i,k}^{- })^{-1}\hat{x}_{i,k}^- + \tilde{C}_i^T \tilde{R}_i^{\dagger}\tilde{y}_{i,k}\right),
		\end{aligned}
	\end{equation}
	where the pseudoinverse is used since $\tilde{R}_i$ might not be invertible.
	
	\subsection{Mismatched Covariance in Existing Algorithms}
		In this section, we will provide a brief review of the CM, CI and HCMCI, and discuss how the mismatched covariance matrices induce performance degradation. In all three algorithms, the prediction step is the same as \eqref{eq: sec2 CKF predict}, so we will only discuss the differences in the measurement update step.
		
		In CM\cite{qian2022consensus},  the measurement update step is
		\begin{equation*}
			\begin{aligned}
				P_{i,k|k} &\textstyle= \left(P_{i,k|k-1}^{-1}+ \sum_{j=1}^{N}l_{ij}^{(\gamma)}N C_j^T R_j^{-1}C_{j}\right)^{-1},\\
				\hat{x}_{i,k|k} &\textstyle= P_{i,k|k} \left(P_{i,k|k-1}^{-1}\hat{x}_{i,k|k-1} + 	\sum_{j=1}^{N}l_{ij}^{(\gamma)}N C_j^T R_j^{-1}y_{j,k} \right).
			\end{aligned}
		\end{equation*}
		By comparing with \eqref{eq:minimum variance estimator}, we can find that CM is numerically equivalent to using the approximation of
			$\tilde{R}_i \approx \frac{1}{N}\tilde{C}_i$.
		In most scenarios this approximation may be inaccurate, leading to a degradation in estimation accuracy. Moreover, since $\frac{1}{N}\tilde{C}_i$ may be smaller than $\tilde{R}_i$, this can cause \( P_{i,k|k} \) to be smaller than the actual error covariance,  which is known as inconsistency.
		
	Similar problem occurs in  CI\cite{battistelli2014kullback} as well.  The measurement update step in  CI is
	\begin{equation*}
		\label{eq:sec2 CI upadate}
		\begin{aligned}
			P_{i,k|k} &\textstyle= \left(\sum_{j=1}^{N}l_{ij}^{(\gamma)} (P_{j,k|k-1}^{-1}+ C_j^T R_j^{-1}C_{j})\right)^{-1},\\
			\hat{x}_{i,k|k} &\textstyle= P_{i,k|k} \left( \sum_{j=1}^{N}l_{ij}^{(\gamma)}(P_{j,k|k-1}^{-1}\hat{x}_{j,k|k-1} +  C_j^T R_j^{-1}y_{j,k} )\right).
		\end{aligned}
	\end{equation*}
	Focusing on the fused measurement, we can find that CI is numerically equivalent to using the approximation 
		 $\tilde{R}_i \approx \tilde{C}_i$.	
	Since $\tilde{C}_i \ge \tilde{R}_i$, this approximation is conservative, which helps guarantee the consistency. However, this conservatism also degenerates the performance of CI. Even if $\gamma$ tends to infinity, the gap between $\tilde{C}_i $ and $\tilde{R}_i$ may remain nonzero, and therefore the estimation accuracy of CI cannot approach that of CKF by increasing the number of consensus steps per iteration.

	To overcome this drawback, HCMCI\cite{battistelli2014consensus} brings into a coefficient before the fused measurement, and the measurement update step is modified as
	\begin{equation*}
		\small
		\begin{aligned}
			P_{i,k|k} &\textstyle= \left(\sum_{j=1}^{N}l_{ij}^{(\gamma)} (P_{j,k|k-1}^{-1}+ \omega_{i,k} C_j^T R_j^{-1}C_{j})\right)^{-1},\\
			\hat{x}_{i,k|k} &\textstyle= P_{i,k|k} \left( \sum_{j=1}^{N}l_{ij}^{(\gamma)}(P_{j,k|k-1}^{-1}\hat{x}_{j,k|k-1}+ \omega_{i,k} C_j^T R_j^{-1}y_{j,k})\right).
		\end{aligned}
	\end{equation*}
	If $\omega_{i,k}=N$, it can be derived that HCMCI will converge to CKF as $\gamma$ tends to infinity, but inconsistency will occur when $\gamma$ is finite, as discussed in\cite{battistelli2014consensus}. If  $\omega_{i,k}$ is set to be other static scalars, then the aforementioned asymptotic optimality will not hold for HCMCI.

	\subsection{Problem Formulation}
		From the above discussion, it can be concluded that in order to optimally exploiting the fused measurement $\tilde{y}_{i,k}$, the key task is to obtain the accurate covariance $\tilde{R}_i$.
		However, computing $\tilde{R}_i$ directly from \eqref{eq:covariance of  fuesed measurement} is challenging, as each node requires: 1) the global weight matrix $\mathcal{L}$ to compute the weights $l_{ij}^{(\gamma)}$, and 2)  \( C_j \) and \( R_j \) from all nodes within $\gamma$ steps. To address this, we present two iterative methods to calculate $\tilde{R}_i$ in the next section.
	
		The distributed filtering problem is formulated as follows:
		\begin{enumerate}[1)]
			\item Establish new methods to calculate the accurate covariance matrix $\tilde{R}_i$ in a fully distributed way.
			\item Modify the CM/CI algorithms using the proposed methods.
			\item Provide stability and performance analysis of the modified algorithms.
		\end{enumerate}

	\section{Calculation of the Quadratic Weighted Sum}
		In this section, we propose two methods for estimating \(\tilde{R}_i\): the direct method and the stochastic method. For each method, we present a convergence analysis of its approximation error, including the rate of convergence.

		For brevity, the following QWS is considered:
		$$\tilde{X}_i := \sum_{j=1}^{N}(l_{ij}^{(\gamma)})^2 X_j,$$
		where $X_j$ can be any positive semi-definite matrix. By setting $X_j = C_j^TR_j^{-1}C_j$, we obtain $\tilde{R}_i$. 
		
		\subsection{Direct Method}
			Since $X_i\ge 0$, we can always find $Y_i\in \mathbb{R}^{n\times n}$ satisfying $X_i = Y_i^T Y_i$. Using this decomposition, we can rewrite $\tilde{X}_i$ as 
			\begin{equation}
				\label{eq:sec3 direct decomposition}
			\tilde{X}_i = \bar{Y}_i ^T \bar{Y}_i,
			\end{equation}
			where $\bar{Y}_i:= \col_{i=1}^N(l_{ij}^{(\gamma)}Y_j)\in \mathbb{R}^{Nn\times n}$.
			
			Assume that each node possesses an individual row vector $q_i\in \mathbb{R}^{1\times N}$, satisfying that  $\mathbi{Q}:=\col_{i=1}^N(q_i)$ is of full rank.

			Then we have 
			\begin{equation*}
				(\mathbi{Q}\otimes I_{n}) \left((\mathbi{Q}^T \mathbi{Q})^{-1} \otimes I_{n}\right)(\mathbi{Q}^T\otimes I_{n}) = I_{Nn}.
			\end{equation*}
			
			By introducing this equation into \eqref{eq:sec3 direct decomposition} and denoting 
			$$\tilde{u}_i := \sum_{j=1}^N l_{ij}^{(\gamma)}Y_j^T (q_j\otimes I_n),$$ 
			it follows that
			\begin{equation}
				\label{eq:eq:sec3 direct decomposition2}
				\begin{aligned}
				\bar{Y}_i ^T \bar{Y}_i & = \bar{Y}_i ^T(\mathbi{Q}\otimes I_{n}) \left[(\mathbi{Q}^T \mathbi{Q})^{-1} \otimes I_{n}\right](\mathbi{Q}^T\otimes I_{n})  \bar{Y}_i\\
				&=\tilde{u}_i \left(\sum_{j=1}^{N}q_j^Tq_j\otimes I_n\right)^{-1} \tilde{u}_i^T.
				\end{aligned}
			\end{equation}
			
			Note that the item $\sum_{j=1}^{N}q_j^Tq_j$ can be calculated with well established average consensus technique \cite{xiao2005scheme}
			\begin{equation}
				\label{eq:direct method gamma infty}
				\lim\limits_{k \rightarrow \infty} \sum_{j=1}^{N} N l_{ij}^{(k)}q_j^Tq_j = \sum_{j=1}^{N} q_j^Tq_j,
			\end{equation}
			and $\tilde{u}_{i,k}$ can be directly obtained after $\gamma$ times of fusion, following rules similar to  \eqref{Sec2C_eq3y0}\eqref{Sec2C_eq4yfusion}.
			Therefore, $\tilde{X}_i$ can be obtained.
			
		As discussed above, the row vector $q_i$ plays an important part in this method. If each node can obtain its local $q_i$ and the global decoupled matrix $\sum_{j=1}^{N}q_j^Tq_j$ offline, it can use the decomposition \eqref{eq:eq:sec3 direct decomposition2} to calculate the desired QWS. When offline assignment is not available, a feasible choice of $q_i$ is to let each entry of it be sampled from independent identical normal distributions
			$[q_i]_j\sim \mathbf{N}(0,1), j\in \underline{N}$, 
		then the augmented matrix $\mathbi{Q}$ is full rank with probability 1 \cite{rudelson2010non}. 
		
		It is worth mentioning that when $k$ is small, the invertibility of $ \sum_{j=1}^{N} N l_{ij}^{(k)}q_j^Tq_j$ may not be guaranteed, thus we use the Moore-Penrose inverse instead of matrix inverse in \eqref{eq:eq:sec3 direct decomposition2}. Denote the estimate of $\tilde{R}_i$ as 
		$$\tilde{U}_{i,k} :=  \tilde{u}_i  \left[( \sum_{j=1}^{N}l_{ij}^{(k)}N q_j^Tq_j )\otimes I_n\right]^\dagger   (\tilde{u}_i )^T,$$ we have the following results for the estimation error.
		
		\begin{lemma}
			\label{lemma:direct method error}
			If \rm{$\mathbi{Q}$} is full rank, then for any $\gamma \in \mathbb{Z}^+$,
				\begin{enumerate}[1)]
				\item $\tilde{U}_{i,k} \rightarrow \tilde{X}_{i}, \tilde{U}_{i,k}^\dagger \rightarrow \tilde{X}_{i}^\dagger$, as $k\rightarrow \infty$;
				\item 	$	\| \tilde{U}_{i,k} - \tilde{X}_{i}\|_2 \le \alpha_i \|\tilde{X}_{i}\|_2$, 
			where  $\alpha_i = \max_{j \in \mathcal{J}_i^{(\gamma)} } |\frac{1}{Nl_{ij}^{(k)}}-1|$, $\mathcal{J}_i^{(\gamma)} = \{j\in \mathcal{N}| l_{ij}^{(\gamma)}>0\}$, and $k \ge \gamma$.
			
			98Moreover, let \(\lambda_2\) denote the eigenvalue of \(\mathcal{L}\) with the second largest absolute value, and $k_0\in\mathbb{Z}^+$ satisfy $k_0 > \log_{|\lambda_2|}(\frac{1}{N})$. Then for any $k \ge k_0$, $\alpha_i $ can be set as $\frac{N|\lambda_2|^k}{1-N|\lambda_2|^{k}}$.
			\end{enumerate}
		\end{lemma}
		
		\begin{proof}
			See Appendix A.
		\end{proof}
		
		\begin{remark}
			\label{remark:direct method dimension}
			Note that the dimension of $q_i$ is dependent on the size of the network, which may lead to heavy communication burden with large network. In order to solve this problem, we provide a stochastic method in Section III.C, which significantly reduces the communication cost.
			
			Meanwhile, in the case that some of the matrices $X_i$ are zero matrices,  we can modify the direct method as
			\begin{equation*}
				\left\{
				\begin{aligned}
					&[q_i]_j \sim \mathbf{N}(0,1) &  i \in \mathcal{S}, &\\
					&[q_i]_j= 0 &  i \notin \mathcal{S}, &\quad j\in \underline{S},
				\end{aligned}\right.
			\end{equation*}
			where $\mathcal{S}$ denotes the set of nodes whose $X_i\neq \boldsymbol{O}$, and $S$ denotes the number of nodes in $\mathcal{S}$. This case corresponds to the scenario that some of the nodes in the network are naive, meaning they can only communicate and process data, and thus have $C_j^TR_j^{-1}C_j=\boldsymbol{O}$. In this situation, \eqref{eq:eq:sec3 direct decomposition2} can be modified as
			\begin{equation*}
				\begin{aligned}
					\bar{Y}_i ^T \bar{Y}_i 
					&=\tilde{u}_i^\mathcal{S}(\sum_{j\in \mathcal{S}}q_j^Tq_j\otimes I_n)^{-1} (\tilde{u}_i^\mathcal{S})^T,
				\end{aligned}		
				\end{equation*}
			where $\tilde{u}_i^\mathcal{S} =  \sum_{j\in \mathcal{S}} l_{ij}^{(\gamma)}Y_j^T (q_j\otimes I_n)$.
			Through this way, the direct method only requires $\col_{i\in \mathcal{S}}(q_i)$ to be invertible, and the minimum dimension of $q_i$ is reduced from $N$ to $S$.  
			
		\end{remark}

		\subsection{Stochastic Method}
		As discussed in Remark~\ref{remark:direct method dimension}, the problem of computing and communicating inefficiency of direct method will become prominent as the scale of the network increases. In direct method, the core of decoupling relies on the following relation
		\begin{equation*}
			q_i \left(\sum_{j=1}^{N}q_j q_j^T\right)^{-1}q_k^T = \delta_{ik},
		\end{equation*}
		as described in the proof of Lemma~\ref{lemma:direct method error}. Our idea here is to substitute it with stochastic relation. Let $\theta_i$ be sampled from identical independent  distributions $\mathbf{N}(\boldsymbol{0}, I_n)$, then it follows that
		\begin{equation}
			\label{eq:stochastic method core}
			\mathbb{E}[\theta_i\theta_j^T] = \delta_{ij}I_n.
		\end{equation}
		Exploiting this, it can be verified that
		\begin{equation}
			\label{eq:stochastic method eq2}
			\begin{aligned}
				&\phantom{=}\mathbb{E}\left[(\sum_{j=1}^N l_{ij}^{(\gamma)}Y_j^T \theta_j)(\sum_{j=1}^N l_{ij}^{(\gamma)}Y_j^T \theta_j)^T\right]
				= \tilde{X}_i.
			\end{aligned}
		\end{equation}

		Through the above decomposition, we only need to take a sufficiently large number of samples  of 
			$\sum_{j=1}^N l_{ij}^{(\gamma)}Y_j^T \theta_j $ and estimate $\tilde{X}_i$ by
		\begin{equation}
			\label{eq:stochastic method kappa infty}
			\begin{aligned} \frac{1}{k}\sum_{t=1}^{k}\left[(\sum_{j=1}^{N}l_{ij}^{(\gamma)}Y_j^T \theta_{j,t}) (\sum_{j=1}^{N}l_{ij}^{(\gamma)}Y_j^T \theta_{j,t})^T \right]
			\end{aligned}
		\end{equation}
		where $\theta_{i,k}$ denotes the $k$-th sample taken from $\mathbf{N}(0,I_n)$.
		
		\begin{remark}
			\label{remark:stochastic method}
		Let $X_j=C_j^TR_j^{-1}C_j$. Note that $Y_j^T\theta_j$ has the same distribution with $C_j^T R_j^{-1} v_j$. From this perspective, the stochastic method can be regarded as that each node provides a sample of the measurement noise and then fuse it. Subsequently, the covariance of the fused measurement noise is estimated statistically, exploiting the fused  samples of $Y_j^T\theta_j$.
		\end{remark}

		\subsection{Properties of the Stochastic Method}		
		For brevity, we define the following variables
			\begin{equation}
				\label{eq:stochastic method eq1}
			 \tilde{\upsilon}_{i,k}:=\sum_{j=1}^{N}l_{ij}^{(\gamma)}Y_j^T\theta_{j,k}, \;\,
			\tilde{\Upsilon}_{i,k} := \frac{1}{k}\sum_{s=1}^{k}  \tilde{\upsilon}_{i,s}  \tilde{\upsilon}_{i,s}^T.
		\end{equation}
		
		Since the auxiliary variable $\tilde{\Upsilon}_{i,k}$ is a random matrix, we can only determine its convergent behavior from the statistical perspective. In this paper, we use the following definition of almost surely convergence of random matrices.
		
		\begin{definition}
			\label{definition:AlmostSureConvergenceMatrix}
			A sequence of random matrices $\{E_k\}_{k =1}^{\infty}$ is said to converge almost surely to $E$ if $\mathbb{P}(\lim_{n\rightarrow \infty} \|E_n - E\|_F = 0) = 1$, and denoted as $E_k \stackrel{a.s.}{\longrightarrow} E$.
		\end{definition}
		
		For almost surely convergent random matrices, the operations of matrix multiplication and inversion  preserve the almost surely convergent property, as stated in the following lemma.
			
		\begin{lemma}
			\label{lemma: a.s. mat prod mat inv}
			\begin{enumerate}
				\item If $D_k \stackrel{a.s.}{\longrightarrow} D$ and $E_k\stackrel{a.s.}{\longrightarrow} E$, then $D_kE_k \stackrel{a.s.}{\longrightarrow} DE$.
				\item If $D_k \stackrel{a.s.}{\longrightarrow} D$ and $D$ is invertible, then $D_k^\dagger  \stackrel{a.s.}{\longrightarrow} D^{-1}$.
			\end{enumerate}
		\end{lemma} 
		
		\begin{proof}
			See Appendix B.
		\end{proof}
		
		Based on the above tools, we can now describe the properties of $\tilde{\Upsilon}_{i,k}$ and $\tilde{\Upsilon}_{i,k}^{\dagger}$.  
		
		\begin{lemma}
			\label{lemma:stochastic method error}
			For any $\gamma \in \mathbb{Z}^+$ and $k>n+3$,
			\begin{enumerate}[1)]
				\item 	$ \tilde{\Upsilon}_{i,k} \stackrel{a.s.}{\longrightarrow}\tilde{X}_i $,
				$\tilde{\Upsilon}_{i,k}^{\dagger} \stackrel{a.s.}{\longrightarrow}\tilde{X}_i^{\dagger}$, as $k\rightarrow\infty$;
				
				\item  if $\tilde{X}_i>0$, then $\tilde{\Upsilon}_{i,k} \sim W_n(\frac{1}{k}\tilde{X}_i, k)$, $\tilde{\Upsilon}_{i,k}^{-1} \sim W_n^{-1}(\frac{1}{k}\tilde{X}_i, k)$;
				\item 	$\textstyle \mathbb{E}[\tilde{\Upsilon}_{i,k}] = \tilde{X}_i$, $\mathbb{E}[\tilde{\Upsilon}_{i,k}^{\dagger}] = \frac{k}{k -r_i -1}\tilde{X}_i^{\dagger} $;
				\item  $ \mathbb{E}[\|\tilde{\Upsilon}_{i,k}- \tilde{X}_i\|_F^2 ]=  \textstyle \frac{1}{k} \operatorname{Tr}( \tilde{X}_i^2) + \frac{1}{k} [\operatorname{Tr}( \tilde{X}_i)]^2 $, 
				
				$  \mathbb{E}[\|  \tilde{\Upsilon}_{i,k}^{\dagger} -\tilde{X}_i^{\dagger} \|_F^2 ]= \alpha_{1,i}  \operatorname{Tr}[(\tilde{X}_i^{\dagger})^2] + \alpha_{2,i}  [\operatorname{Tr}(\tilde{X}_i^{\dagger})]^2$,
				
				where
				
				$ \alpha_{1,i} = \frac{k^2 + k (r_i^2 + 2 r_i + 3 ) - (r_i^3 + 4 r_i^2 + 3 r_i)}{(k - r_i - 3) (k - r_i - 1) (k - r_i) }$, 
				
				$\alpha_{2,i} = \frac{k^2}{(k - r_i - 3) (k - r_i - 1) (k - r_i) } $;
			\end{enumerate}
			where $r_i := \text{rank}(\tilde{X}_i)$, $ W_n(\Sigma, p)$ and $ W_n^{-1}(\Sigma, p)$ denotes the $n$-dimension Wishart distribution and Inverse Wishart distribution, respectively, with degree of freedom $p$ and scale matrix $\Sigma$.
		\end{lemma}
		
		\begin{proof}
			See Appendix B.
		\end{proof}
		
		\begin{remark}
			The stochastic method has a lower communication cost than the direct method.
			However, it only performs well when dealing with  static QWS. If the required QWS is time-varying, i.e, $\tilde{X}_{i,k} = \sum_{j=1}^{N}(l_{ij}^{(\gamma)})^2 Y_{j,k}^TY_{j,k}$, it will be  inefficient  to still use \eqref{eq:stochastic method kappa infty}, as it requires to take large number of samples at each time instant. In contrast, whether addressing static  or time-varying QWS, $q_i$ in the direct method remains static, and thus $(\sum_{j=1}^{N}q_j^Tq_j)^{-1}$ remains static as well. The only adjustment  when shifting from static QWS to dynamic QWS, is to replace $Y_{j}$ by   $Y_{j,k}$ in \eqref{eq:eq:sec3 direct decomposition2}. It is worth mentioning that such adjustment does not increase  the cost of  computation and communication.
		\end{remark}
						
		\section{Distributed Kalman Filter with Ultimately Accurate Covariance}
		\subsection{The Modified CM/CI}
		Under the guidance of the aforementioned idea, the Modified CM is depicted in Algorithms~1-2, and the Modified CI is shown in Algorithms~3-4. Note that in the fusion steps of Algorithms~1-4, \( Z \) represents a general term, indicating that the subsequent quantities are fused according to the same rule.
		
		\subsection{Essential Discussions}
		
	\begin{remark}
		Compared with CM and CI, the modifications in Algorithms~1-4 lie in the use of the ultimately accurate measurement and covariance pair $(\tilde{y}_{i,k},\tilde{U}_{i,k})$ or $(\tilde{y}_{i,k},\tilde{\Upsilon}_{i,k})$ in the measurement update step, where $\lim_{k\rightarrow \infty} \tilde{U}_{i,k}=\lim_{k\rightarrow \infty} \tilde{\Upsilon}_{i,k} = \tilde{R}_i$, while CM and CI use inherently inaccurate pairs, $(\tilde{y}_{i,k}^{(\gamma)}, \frac{1}{N}\tilde{C}_i)$ and $(\tilde{y}_{i,k}^{(\gamma)}, \tilde{C}_i)$, respectively. 
	\end{remark}
	
	\begin{remark}
		Algorithms~1-4 have different requirements for global information during initialization. Specifically, each node needs to know the system transition matrix $A$, the process noise covariance $Q$, and the distribution of the initial state estimate $\mathcal{N}(\hat{x}_0,P_0)$. Additionally, in Algorithm~1 and Algorithm~3, each node should also be aware of the network size $N$. Apart from the global information, the $i$-th node should also know its own observation matrix $C_i$, the measurement noise covariance $R_i$, and the consensus weights $\{l_{ij}:j\in\mathcal{N}_i\}$. 
	\end{remark}
	
	\begin{remark}
			To compare communication costs, let \(\tau\) denote the number of scalars sent by one single node at each fusion step. Then, one has $\tau_{CM} = \tau_{CI} = n^2+n$, $\tau_{HCMCI} = 2n^2+2n$, $\tau_{Alg1} = N^2+Nn^2+n^2+n$, $\tau_{Alg2} = n^2+2n$, $\tau_{Alg3} = N^2+Nn^2+2n^2+2n$, $\tau_{Alg4} = 2n^2+3n$. 
			
			Since the system is LTI, there's no need to update $\tilde{U}_{i,k}$ or $\tilde{\Upsilon}_{i,k}$ when it has converged. If this rule is applied, one has $\tau_{CM} = \tau_{CI} =\tau_{Alg1}  = \tau_{Alg2}  = n$, $ \tau_{HCMCI} = \tau_{Alg3}  = \tau_{Alg4} = 2n$ in steady state.
	\end{remark}

		\begin{table}[t]
			\centering
			\begin{tabular}{c}
				\hline
				Algorithm~1: Modified CM with direct method\\
				\hline			
			\begin{minipage}{0.95\linewidth}
				\begin{tabular}{ll}
				1. Initialization: & \(\hat{x}_{i,0|0} = x_0, P_{i,0|0} = P_0, \tilde{U}_{i,0} =  \boldsymbol{O}_n\), \\
				&$q_i \in \mathbb{R}^{1\times N}: [q_i]_j\sim \mathbf{N}(0,1),j \in \underline{N}$,\\
				 & $Y_i  \in \mathbb{R}^{n\times n}: C_i^T R_i^{-1}C_i = Y_i^TY_i$,  \\
				 &$\tilde{W}_{i,0}^{(\gamma)} = N q_i^Tq_i$. \\
				2. Prediction: & $\hat{x}_{i,k|k-1} = A \hat{x}_{i,k-1|k-1}, $\\
				&  $ P_{i,k|k-1} = AP_{i,k-1|k-1} A^T + Q $.\\
				3. Fusion: & $	\tilde{y}_{i,k}^{(0)} =  C_i^T R_i^{-1} y_{i,k},\quad	\tilde{u}_{i,k}^{(0)} =  Y_i^T(q_i\otimes I_n),$\\
				&$ S_{i,k}^{(0)} =  C_i^T R_i^{-1} C_i,\quad   \tilde{W}_{i,k}^{(0)} =\tilde{W}_{i,k-1}^{(\gamma)}$. \\
				& For  \(m = 1,2,\dots,\gamma\) \\
				& 	$	Z_{i,k}^{(m)} = \sum_{j=1}^{N}l_{ij}Z_{j,k}^{(m-1)} $,  $Z = \tilde{y}, \tilde{u}, S, \tilde{W}$.\\
				4. Covariance:& $\tilde{U}_{i,k} = \tilde{u}_{i,k}^{(\gamma)}  (\tilde{W}_{i,k}^{(\gamma)}\otimes I_n)^{\dagger} (\tilde{u}_{i,k}^{(\gamma)})^T$,  $\tilde{C}_i = S_{i,k}^{(\gamma)}$.\\
				5. Correction: 
				& $P_{i,k|k} = (P_{i,k|k-1}^{-1} + {\tilde{C}_i^T  \tilde{U}_{i,k}^{\dagger} \tilde{C}_i} )^{-1}$,\\
				& $ \hat{x}_{i,k|k} = P_{i,k|k} 
				(P_{i,k|k-1}^{-1}\hat{x}_{i,k|k-1}  + {\tilde{C}_i^T\tilde{U}_{i,k}^{\dagger} \tilde{y}_{i,k}} ).$
				\end{tabular}
			\end{minipage}\\
			~\\
				\hline

				Algorithm~2: Modified CM with stochastic method\\
				\hline			
				\begin{minipage}{0.95\linewidth}
					\begin{tabular}{ll}
						1. Initialization: & \(\hat{x}_{i,0|0} = x_0, P_{i,0|0} = P_0, \tilde{\Upsilon}_{i,0} =  \boldsymbol{O}_n\). \\& $Y_i  \in \mathbb{R}^{n\times n}: C_i^T R_i^{-1}C_i = Y_i^TY_i$,  \\
						2. Prediction: & Same as Algorithm~1.\\
						3. Fusion: &Take a sample $\theta_{i,k} $ from $\mathbf{N}(0,I_n)$.\\
						& $	\tilde{y}_{i,k}^{(0)} =  C_i^T R_i^{-1} y_{i,k},\quad	\tilde{\upsilon}_{i,k}^{(0)} =  Y_i^T\theta_{i,k},$\\ 
						&$ S_{i,k}^{(0)} =  C_i^T R_i^{-1} C_i$. \\
						& For  \(m = 1,2,\dots,\gamma\) \\
						& 	$	Z_{i,k}^{(m)} = \sum_{j=1}^{N}l_{ij}Z_{j,k}^{(m-1)} $, $Z = \tilde{y}, \tilde{\upsilon}, S$.\\
						4. Covariance: &  $\tilde{\Upsilon}_{i,k} = \frac{k-1}{k} \tilde{\Upsilon}_{i,k-1}  + \frac{1}{k}\tilde{\upsilon}_{i,k}^{(\gamma)} (\tilde{\upsilon}_{i,k}^{(\gamma)})^T$,\;$\tilde{C}_i = S_{i,k}^{(\gamma)}$.\\
						5. Correction: & Same as Algorithm~1.
					\end{tabular}
				\end{minipage}\\
				~\\
				\hline
				Algorithm~3: Modified CI with direct method\\
				\hline			
				\begin{minipage}{0.95\linewidth}
					\begin{tabular}{ll}
						1. Initialization: & Same as Algorithm~1. \\
						2. Prediction: & Same as Algorithm~1. \\
						3. Fusion: 
						& $	\tilde{y}_{i,k}^{(0)} =  C_i^T R_i^{-1} y_{i,k},\quad	\tilde{u}_{i,k}^{(0)} =  Y_i^T(q_i\otimes I_n),$\\
						&$ S_{i,k}^{(0)} =  C_i^T R_i^{-1} C_i,\quad   \tilde{W}_{i,k}^{(0)} =\tilde{W}_{i,k-1}^{(\gamma)}$, \\
						&$J_{i,k}^{(0)} = P_{i,k|k-1}^{-1}\hat{x}_{i,k|k-1}, \quad \tilde{V}_{i,k}^{(0)} =  P_{i,k|k-1}^{-1}.$ \\
						& For  \(m = 1,2,\dots,\gamma\) \\
						& 	$	Z_{i,k}^{(m)} = \sum_{j=1}^{N}l_{ij}Z_{j,k}^{(m-1)} $,  $Z = \tilde{y}, \tilde{u}, S, \tilde{W}, J, \tilde{V}$.\\
						4. Covariance: &Same as Algorithm~1. \\
						5. Correction: 
						& $ P_{i,k|k} = (\tilde{V}_{i,k}^{(\gamma)} + {\tilde{C}_i^T  \tilde{U}_{i,k}^{\dagger} \tilde{C}_i} )^{-1}$,\\
						& $ \hat{x}_{i,k|k} = P_{i,k|k} 
						( J_{i,k}^{(\gamma)} + {\tilde{C}_i^T\tilde{U}_{i,k}^{\dagger} \tilde{y}_{i,k}} ).$
						
					\end{tabular}
				\end{minipage}\\
				~\\
				\hline
				Algorithm~4: Modified CI with stochastic method\\
				\hline			
				\begin{minipage}{0.95\linewidth}
					\begin{tabular}{ll}
						1. Initialization:  & Same as Algorithm~2,\\
						2. Prediction: & Same as Algorithm~1,\\
						3. Fusion: &Take a sample $\theta_{i,k} $ from $\mathbf{N}(0,I_n)$.\\
						& $	\tilde{y}_{i,k}^{(0)} = C_i^T R_i^{-1} y_{i,k},\quad	\tilde{\upsilon}_{i,k}^{(0)} = Y_i^T\theta_{i,k}$,\\
						&$ S_{i,k}^{(0)} = C_i^T R_i^{-1} C_i$ \\
						&$J_{i,k}^{(0)} = P_{i,k|k-1}^{-1}\hat{x}_{i,k|k-1}, \quad \tilde{V}_{i,k}^{(0)} =  P_{i,k|k-1}^{-1}$. \\
						& For  \(m = 1,2,\dots,\gamma\) \\
						& 	$	Z_{i,k}^{(m)} = \sum_{j=1}^{N}l_{ij}Z_{j,k}^{(m-1)} $,  $Z = \tilde{y}, \tilde{\upsilon}, S, J,\tilde{V} $.\\
						4. Covariance: &Same as Algorithm~2.\\
						5. Correction: & Same as Algorithm~3.
					\end{tabular}
				\end{minipage}\\
				\hline
			\end{tabular}
		\end{table}

	\subsection{Stability Analysis}
	Let $\mathsf{P}_0,\mathsf{P}_{i,0}, \mathsf{R} ,\mathsf{R}_i, \mathsf{Q}>0$, $\mathsf{A}, \mathsf{C}, \mathsf{C}_i$ be matrices with appropriate dimension. Consider the following DARE
	\begin{equation*}
		\textstyle \mathsf{P} = \mathsf{A}\left(\mathsf{P}^{-1}+ \mathsf{C}^T \mathsf{R}^{-1}\mathsf{C}\right)^{-1}\mathsf{A}^T+\mathsf{Q},
	\end{equation*}
	the modified HCRE
	\begin{equation*}
	\textstyle	\mathsf{P}_i = \mathsf{A}\left(\sum_{j=1}^{N} l_{ij}\mathsf{P}_j^{-1}+ \mathsf{C}_i^T \mathsf{R}_i^{-1}\mathsf{C}_i\right)^{-1}\mathsf{A}^T+\mathsf{Q},
	\end{equation*}
	and the corresponding DARR
	\begin{equation*}
	\textstyle	\mathsf{P}_{k+1} = \mathsf{A}\left(\mathsf{P}_{k}^{-1}+ \mathsf{C}^T \mathsf{R}^{-1}\mathsf{C}\right)^{-1}\mathsf{A}^T+\mathsf{Q}, 
	\end{equation*}
	and HCRR
	\begin{equation*}
	\textstyle	\mathsf{P}_{i,k+1} = \mathsf{A}\left(\sum_{j=1}^{N} l_{ij}\mathsf{P}_{j,k}^{-1}+ \mathsf{C}_i^T \mathsf{R}_i^{-1}\mathsf{C}_i\right)^{-1}\mathsf{A}^T+\mathsf{Q}.
	\end{equation*}

	It is shown in the literature that these equations are important for depicting the steady-state performance of distributed filtering algorithms. We 
	 have the following results for these equations and recursions.
	\begin{lemma}
	\label{lemma:DARE-HCRE-Existence-Uniqueness-Convergence}
		\begin{enumerate}[1)]
		\item If $(\mathsf{A},\mathsf{C})$ is observable, then the DARE has a unique solution, and the DARR converges to the solution of DARE, regardless of the initial value $\mathsf{P}_0$.
		\item If $\mathsf{A}$ is invertible and $(\mathsf{A},\col_{i=1}^N(\mathsf{C}_i))$ is observable, then the HCRE has a unique solution, and the HCRR converges to the solution of HCRE, regardless of the initial values $\mathsf{P}_{i,0}$.
		\end{enumerate}
	\end{lemma}
	\begin{proof}
		{1)} can be found in \cite{anderson1979optimal}.
		
		{2)}
		Firstly, we prove that $\mathsf{P}_{i,k}$ is uniformly bounded. This proof is  modified from the proof of Theorem~4 in \cite{battistelli2014kullback}.
		
		It is obvious that $\mathsf{P}_{i,k} >\mathsf{Q}$, thus we have 
		\begin{equation*}
			\sum_{j=1}^{N}l_{ij}\mathsf{P}_{j,k}^{-1} + \mathsf{C}_i^T \mathsf{R}_i^{-1}\mathsf{C}_i \le \mathsf{Q}^{-1}+ \sum_{j=1}^{N}\mathsf{C}_j^{T} \mathsf{R}_j^{-1}\mathsf{C}_j.
		\end{equation*}
		By fact (ii) of Lemma~1 in \cite{battistelli2014kullback}, for some positive scalar  $\check{\beta}\le 1$, the following inequality holds
		\begin{equation*}
			\begin{aligned}
			\mathsf{P}_{i,k+1}^{-1} 
			&\textstyle= ( \mathsf{A}(\sum_{j=1}^{N}l_{ij}\mathsf{P}_{j,k}^{-1} + \mathsf{C}_i^T \mathsf{R}_i^{-1}\mathsf{C}_i )^{-1}\mathsf{A}^T +\mathsf{Q})^{-1}\\
			&\textstyle\ge \check{\beta}\mathsf{A}^{-T} (\sum_{j=1}^{N}l_{ij}\mathsf{P}_{j,k}^{-1} + \mathsf{C}_i^T \mathsf{R}_i^{-1}\mathsf{C}_i) \mathsf{A}^{-1} \\
			&\textstyle \ge \check{\beta}^{\bar{k}}(\mathsf{A}^{-T})^{\bar{k}}(\sum_{j=1}^{N}l_{ij}^{(\bar{k})}\mathsf{P}_{j,k+1-\bar{k}}^{-1} )(\mathsf{A}^{-1})^{\bar{k}} +\tilde{\Omega}_i,
			 \end{aligned}
		\end{equation*}
		where $\tilde{\Omega}_i  = \sum_{s=2}^{\bar{k}}\check{\beta}^{s}(\mathsf{A}^{-T})^{s}(\sum_{j=1}^{N}l_{ij}^{(s-1)}\mathsf{C}_j^T\mathsf{R}_j^{-1}\mathsf{C}_j )(\mathsf{A}^{-1})^{s}$, $\bar{k}\ge 2$.
		
		As discussed in \cite{battistelli2014kullback}, $\tilde{\Omega}_i$ is positive definite for sufficiently large $\bar{k}$, thus $ \mathsf{P}_{i,k+1} \le \tilde{\Omega}^{-1}$ for $k\ge \bar{k}$. Together with $\mathsf{P}_{i,k} >\mathsf{Q}$, we have $\mathsf{P}_{i,k+1}$ is uniformly bounded.
		
		By following the same procedure used in the proofs of Lemmas~2-3, and Theorems~1-2 in \cite{qian2022harmonic}, with replacing $\tilde{C}_i$ and $\tilde{R}_i$ by  $\mathsf{C}_i$ and $ \mathsf{R}_i$, both the existence and uniqueness of the solution to the HCRE, as well as the convergence of the HCRR to the solution of the HCRE, can be established.

	\end{proof}
	
	\begin{lemma}~
		\label{lemma:DARE-HCRE-preserve-order}
	\begin{enumerate}[1)]
		\item If $(\mathsf{A},\mathsf{C})$ is observable, then the solution of the $DARE$ preserves the order of $\mathsf{Q}, \mathsf{R}$.
		\item  If $\mathsf{A}$ is invertible and $(\mathsf{A},\col_{i=1}^N(\mathsf{C}_i))$ is observable, then the solution of the $HCRE$  preserves the order of $\mathsf{Q}, \diag_{i=1}^N(\mathsf{R}_i)$.
	\end{enumerate}
	\end{lemma}
	
	\begin{proof}
		
		{1)}
		Let $\mathsf{Q}_1\ge \mathsf{Q}_2$ and $\mathsf{R}_1 \ge \mathsf{R}_2$. We shall prove $\mathsf{P}_1 \ge \mathsf{P}_2$, where $\mathsf{P}_i$ is the solution to $DARE(\mathsf{A}, \mathsf{C}, \mathsf{Q}_i, \mathsf{R}_i)$.
		
		We consider the corresponding DARR firstly. Assume that $\mathsf{P}_{1,0} \ge \mathsf{P}_{2,0}$. For all $k\in \mathbb{Z}^+$, if $\mathsf{P}_{1,k-1} \ge \mathsf{P}_{2,k-1}$, we have 
		\begin{equation*}
			\begin{aligned}
			\mathsf{P}_{1,k} &= \mathsf{A}(\mathsf{P}_{1,k-1}^{-1} + \mathsf{C}^T \mathsf{R}_1^{-1} \mathsf{C})^{-1}\mathsf{A}^T + \mathsf{Q}_1 \\
			&\ge \mathsf{A}(\mathsf{P}_{2,k-1}^{-1} + \mathsf{C}^T \mathsf{R}_2^{-1} \mathsf{C})^{-1}\mathsf{A}^T + \mathsf{Q}_2 = \mathsf{P}_{2,k}.
			\end{aligned}
		\end{equation*}
		Hence, $\mathsf{P}_{1,k} \ge \mathsf{P}_{2,k}$ for all $k\in \mathbb{Z}_0^+$ by deduction. 
		
		Note that $\mathsf{P}_{i,k}$ converges to $\mathsf{P}_{i}$, and $\mathsf{P}_{i}$ is the unique solution to the DARE, the proof is complete.
		
		{2)}
		Similar to 1).
	\end{proof}

	\begin{lemma}
		\label{lemma:DARE-HCRE-Continuous-on-R}
		\begin{enumerate}[1)]
			\item  If $(\mathsf{A},\mathsf{C})$ is observable, then the solution of the $DARE$ is continuously dependent on $\mathsf{R}$, i.e., $ \forall \varepsilon >0$, $\exists \delta >0$, s.t.  $\forall \grave{\mathsf{R}}$ satisfying $\grave{\mathsf{R}} = \grave{\mathsf{R}}^T$ and  $ \|\grave{\mathsf{R}}-\mathsf{R}\|_F< \delta$, the solution  of $DARE({\mathsf{A}}, \mathsf{C}, {\mathsf{Q}}, \grave{\mathsf{R}})$, denoted as $\grave{\mathsf{P}}$, satisfies
			\begin{equation*}
				\|\grave{\mathsf{P}} - \mathsf{P}\|_F < \varepsilon.
			\end{equation*}
			\item If $\mathsf{A}$ is invertible and $(\mathsf{A},\col_{i=1}^N(\mathsf{C}_i))$ is observable, then the solution of the $HCRE$ is continuously dependent on $\diag_{i=1}^N(\mathsf{R}_i)$. 
		\end{enumerate}
	\end{lemma}
		
	\begin{proof} 
		For all $k>0$, it can be easily verified that if $\mathsf{P}$ is the solution to $DARE(\mathsf{A}, \mathsf{C}, \mathsf{Q}, \mathsf{R})$, then $k\mathsf{P}$ is the solution to $DARE(\mathsf{A}, \mathsf{C}, k\mathsf{Q}, k\mathsf{R})$. Similarly, if $\{\mathsf{P}_i\}$ is the solution to $HCRE(\mathsf{A}, \col_{i=1}^N(\mathsf{C}_i), \mathsf{Q}, \diag_{i=1}^N(\mathsf{R}_i), \mathcal{L})$, then $\{k\mathsf{P}_i\}$ is the solution to $HCRE(\mathsf{A}, \col_{i=1}^N(\mathsf{C}_i), k\mathsf{Q}, \diag_{i=1}^N(k\mathsf{R}_i),  \mathcal{L})$.
		
		{1)}
		Given any $\varepsilon>0$, let $\delta =\frac{\alpha \lambda_{min}(\mathsf{R})\varepsilon}{\|\mathsf{P}\|_2}$, where $ 0 < \alpha< \min(\frac{1}{\sqrt{n}}, \frac{\|\mathsf{P}\|_2}{\varepsilon}  )$, and $n$ is the dimension of $\mathsf{A}$.
		
		Then for all $\grave{\mathsf{R}}=\grave{\mathsf{R}}^T$ satisfying $ \|\grave{\mathsf{R}}-\mathsf{R}\|_F< \delta$, 
		\begin{equation*}
			\begin{aligned}
			\mathsf{R}-\frac{\alpha \lambda_{min}(\mathsf{R})\varepsilon}{\|\mathsf{P}\|_2}I 
			< \grave{\mathsf{R}} < \mathsf{R}+\frac{\alpha \lambda_{min}(\mathsf{R})\varepsilon}{\|\mathsf{P}\|_2}I.
			\end{aligned}
		\end{equation*}
		In use of $\lambda_{min}(\mathsf{R})I \le \mathsf{R}$ and $ \frac{\alpha \varepsilon}{\|\mathsf{P}\|_2} < 1$, we have
		\begin{equation*}
			\begin{aligned}
			0 < (1-\frac{\alpha \varepsilon}{\|\mathsf{P}\|_2})\mathsf{R}
			< \grave{\mathsf{R}} < (1+\frac{\alpha \varepsilon}{\|\mathsf{P}\|_2})\mathsf{R}.
		\end{aligned}
		\end{equation*}
		
		Since $(1-\frac{\alpha \varepsilon}{\|\mathsf{P}\|_2})\mathsf{P}$ is the solution to $DARE(\mathsf{A},\mathsf{C},(1-\frac{\alpha \varepsilon}{\|\mathsf{P}\|_2})\mathsf{Q},(1-\frac{\alpha \varepsilon}{\|\mathsf{P}\|_2})\mathsf{R})$ and $\grave{\mathsf{P}}$ is the solution to $DARE(\mathsf{A}, \mathsf{C}, \mathsf{Q}, \grave{\mathsf{R}})$, in combination with Lemma~\ref{lemma:DARE-HCRE-preserve-order}, we have
		\begin{equation*}
			(1-\frac{\alpha \varepsilon}{\|\mathsf{P}\|_2})\mathsf{P} < \grave{\mathsf{P}}.
		\end{equation*}
		It  can be proved in the same way that $ \grave{\mathsf{P}}<(1+\frac{\alpha \varepsilon}{\|\mathsf{P}\|_2})\mathsf{P} $, and we get
		\begin{equation*}
			-\frac{\alpha \varepsilon}{\|\mathsf{P}\|_2}\mathsf{P} < \grave{\mathsf{P}} -\mathsf{P} < \frac{\alpha \varepsilon}{\|\mathsf{P}\|_2}\mathsf{P}.
		\end{equation*}
		Since $\mathsf{P}, \grave{\mathsf{P}}>0$, we have
		\begin{equation*}
			\| \grave{\mathsf{P}} -\mathsf{P} \|_F  \le \sqrt{n} \| \grave{\mathsf{P}} -\mathsf{P} \|_2 <\sqrt{n} \frac{\alpha \varepsilon}{\|\mathsf{P}\|_2}\|\mathsf{P}\|_2< \varepsilon.
		\end{equation*}
		
		{2)} Similar to 1).
	\end{proof}

	\begin{lemma}~ 
		\label{lemma:DARE-HCRE-Convergent-Parameter}
		\begin{enumerate}[1)]
			\item If $(\mathsf{A},\mathsf{C})$ is observable and $\grave{\mathsf{R}}_k$ converges to $\mathsf{R}$, then the $DARR(\mathsf{A},\mathsf{C},\mathsf{Q},\grave{\mathsf{R}}_k)$ converges to the solution of $DARE(\mathsf{A},\mathsf{C},\mathsf{Q},\mathsf{R})$.
			\item If $\mathsf{A}$ is invertible, $(\mathsf{A},\col_{i=1}^N(\mathsf{C}_i))$ is observable, and $\diag_{i=1}^N(\grave{\mathsf{R}}_{i,k})$ converges to $\diag_{i=1}^N(\mathsf{R}_i)$, then the $HCRR(\mathsf{A},\col_{i=1}^N(\mathsf{C}_i),\mathsf{Q},\diag_{i=1}^N(\grave{\mathsf{R}}_{i,k}), \mathcal{L})$ converges to the solution of $HCRE(\mathsf{A},\col_{i=1}^N(\mathsf{C}_i),\mathsf{Q},
			\diag_{i=1}^N(\mathsf{R}_i), \mathcal{L})$.
		\end{enumerate}
	\end{lemma}
	\begin{proof}
	1) Denote the solution of $DARE(\mathsf{A},\mathsf{C},\mathsf{Q},\mathsf{R})$ as $\mathsf{P}$. Consider the sequence $\grave{\mathsf{P}}_{k+1} = DARR(\mathsf{A},\mathsf{C},\mathsf{Q},\grave{\mathsf{R}}_k)$. 
	
	With Lemma~\ref{lemma:DARE-HCRE-Continuous-on-R}, we can conclude that for all $\varepsilon>0$, there exists $\delta>0$, such that the solutions of  $DARE(\mathsf{A}, \mathsf{C}, \mathsf{Q}, \mathsf{R}-\delta I)$ and $DARE(\mathsf{A}, \mathsf{C}, \mathsf{Q}, \mathsf{R}+\delta I_n)$, denoted as $\check{\mathsf{P}}$ and $\hat{\mathsf{P}}$ respectively, satisfy
	\begin{equation*}
		\|\check{\mathsf{P}}-\mathsf{P}\|_F < \varepsilon/2, \; \|\hat{\mathsf{P}}-\mathsf{P}\|_F < \varepsilon/2.
	\end{equation*}
	
	Since $\grave{\mathsf{R}}_k$ converges to $\mathsf{R}$, there exists $N_1\in \mathbb{Z}^+$, for all $k>N_1$, it holds that $\|\grave{\mathsf{R}}_k-\mathsf{R}\|_F<\delta$, which means 
	\begin{equation*}
		\mathsf{R}-\delta I < \grave{\mathsf{R}}_k < \mathsf{R}+\delta I.
	\end{equation*}

	Construct two iterations $\check{\mathsf{P}}_k, \hat{\mathsf{P}}_k $ by
	\begin{equation*}
	\left\{
	\begin{aligned}
		&\check{\mathsf{P}}_{k+1} = \hat{\mathsf{P}}_{k+1} = DARR(\mathsf{A},\mathsf{C},\mathsf{Q},\grave{\mathsf{R}}_k),&k\le N,\\
		&\check{\mathsf{P}}_{k+1} = DARR(\mathsf{A},\mathsf{C},\mathsf{Q},\mathsf{R}-\delta I),&k> N,\\
		&\hat{\mathsf{P}}_{k+1} = DARR(\mathsf{A},\mathsf{C},\mathsf{Q},\mathsf{R}+\delta I),&k> N.
		\end{aligned}\right.
	\end{equation*}
	
	Similarly to the proof of Lemma~\ref{lemma:DARE-HCRE-preserve-order}, it can be obtained that
	\begin{equation}
		\label{eq:proof of lemma8 sandwich}
		\check{\mathsf{P}}_k \le \grave{\mathsf{P}}_k \le \hat{\mathsf{P}}_k, \quad k \in \mathbb{Z}^+. 
	\end{equation}
	
	With Lemma~\ref{lemma:DARE-HCRE-Existence-Uniqueness-Convergence}, we have $\check{\mathsf{P}}_{k}$ and $\hat{\mathsf{P}}_{k}$ converge to $\check{\mathsf{P}}$ and $\hat{\mathsf{P}}$, respectively. Therefore, there exists $N_2>N_1$, for all $k>N_2$, 
	\begin{equation*}
			\|\check{\mathsf{P}}_k-\check{\mathsf{P}}\|_F < \varepsilon/2, \; \|\hat{\mathsf{P}}_k-\hat{\mathsf{P}}\|_F < \varepsilon/2.
	\end{equation*}
	
	Hence, for all $k>N_2$, 
	\begin{equation*}
		\begin{aligned}
				\|\check{\mathsf{P}}_k - \mathsf{P}\|_F \le \|\check{\mathsf{P}}_k-\check{\mathsf{P}}\|_F + \|\check{\mathsf{P}}-\mathsf{P}\|_F =\varepsilon,\\
				\|\hat{\mathsf{P}}_k - \mathsf{P}\|_F \le \|\hat{\mathsf{P}}_k-\hat{\mathsf{P}}\|_F + \|\hat{\mathsf{P}}-\mathsf{P}\|_F =\varepsilon.
		\end{aligned}
	\end{equation*}
	
	Together with \eqref{eq:proof of lemma8 sandwich},  there is 
	\begin{equation*}
		\mathsf{P}-\varepsilon I \le \check{\mathsf{P}}_k \le \grave{\mathsf{P}}_k \le \hat{\mathsf{P}}_k \le \mathsf{P}+\varepsilon I, \; k>N_2, 
	\end{equation*}
	which implies $\|\grave{\mathsf{P}}_k-\mathsf{P}\|_2 \le \varepsilon, \; k>N_2$.
	
	2) Similar to 1).
		\end{proof}

		Based on the previous results for parameter converged DARR and HCRR, we can now analyze the stability and the performance of the Modified CM/CI methods.
		
		\subsection{Performance Analysis of the Modified CM}

		Lemmas~\ref{lemma:DARE-HCRE-Existence-Uniqueness-Convergence}-\ref{lemma:DARE-HCRE-Convergent-Parameter} in the previous subsection guarantee the existence and uniqueness of the solutions to DARE/HCRE and the convergence of  DARR/HCRR, which are useful in analyzing the stability of the proposed methods. However, these lemmas require the parameter $R$ (or $R_i$) to be positive definite, while in the modified methods it can only be guaranteed that  $\tilde{R}_i \ge 0$. Nonetheless, the equations considered in this paper have a special structure such that Lemmas~\ref{lemma:DARE-HCRE-Existence-Uniqueness-Convergence}-\ref{lemma:DARE-HCRE-Convergent-Parameter} can be adapted to the case that $\breve{R}_i$ is not full rank, as stated in the following lemma.
		
		\begin{lemma}
			\label{lemma: R dagger}
			In Algorithms~1-4, there exists  $\breve{R}_i >0$ satisfying $$\tilde{C}_i^T \breve{R}_i^{-1} \tilde{C}_i = \tilde{C}_i^T (\tilde{R}_i)^\dagger \tilde{C}_i.$$
		\end{lemma}
		
		\begin{proof} 
		See Appendix C. 
		\end{proof}
		
		Now we are ready to discuss the convergent results of the modified methods. Note that all the following discussions implicitly assume that Assumptions~1-3 are satisfied.
		
		\begin{theorem}
			\label{theorem:Algorithm-1-2-P-Convergence}
			If $(A, \tilde{C}_i)$ is observable for all $i\in\mathcal{N}$, then 
			\begin{enumerate}[1)]
			\item $P_{i,k+1|k} $ in Algorithm~1 converges to $P_{i}$, 
			\item $P_{i,k+1|k} $ in Algorithm~2 converges almost surely to $P_{i}$, 
			\end{enumerate}
			where $P_i$ is the solution of the following DARE
			\begin{equation}
				\label{eq:Algorithm1-2-steady-state-P}
					P_{i} = A (P_{i}^{-1} + \tilde{C}_i^T \tilde{R}_i^{\dagger} \tilde{C}_i  )^{-1}A^T + Q. 
			\end{equation}
		\end{theorem}

		\begin{proof}
			1) In Algorithm~1, we have
			\begin{equation}
				\label{eq:Theorem1_eq1}
				P_{i,k+1|k} = A (P_{i,k|k-1}^{-1} + \tilde{C}_i^T \tilde{U}_{i,k}^{\dagger} \tilde{C}_i  )^{-1}A^T + Q.
			\end{equation}
			
			By Lemma~\ref{lemma: R dagger}, there exists $\breve{R}_i >0$ such that $\tilde{C}_i^T (\tilde{R}_{i}^{\dagger} - \breve{R}_i^{-1}) \tilde{C}_i = \boldsymbol{O}$. Thus, \eqref{eq:Theorem1_eq1} is equivalent to 
			\begin{equation*}
				P_{i,k+1|k} = A (P_{i,k|k-1}^{-1} + \tilde{C}_i^T (\tilde{U}_{i,k}^{\dagger}- \tilde{R}_{i}^{\dagger} + \breve{R}_i^{-1})  \tilde{C}_i  )^{-1}A^T + Q.
			\end{equation*}
			
			By Lemma~\ref{lemma:direct method error}, $\tilde{U}_{i,k}^\dagger$ converges to $\tilde{R}_i^\dagger$, hence 
			\begin{equation*}
				\lim\limits_{k\rightarrow\infty}(\tilde{U}_{i,k}^{\dagger}- \tilde{R}_{i}^{\dagger} + \breve{R}_i^{-1})^{-1} = \breve{R}_i.
			\end{equation*}
			
			Since $(A, \tilde{C}_i)$ is observable, it follows from Lemma~\ref{lemma:DARE-HCRE-Convergent-Parameter} that $P_{i,k+1|k}$ converges to the solution of 
			\begin{equation}
				\label{eq:theorem1 dare eq1}
				P_{i} = A (P_{i}^{-1} + \tilde{C}_i^T\breve{R}_i^{-1} \tilde{C}_i  )^{-1}A^T + Q,
			\end{equation}
			which is equivalent to \eqref{eq:Algorithm1-2-steady-state-P} by Lemma~\ref{lemma: R dagger}.
			
			2) In Algorithm~2, by Lemma~\ref{lemma:stochastic method error}, we have $\tilde{\Upsilon}_{i,k}^\dagger \stackrel{a.s.}{\longrightarrow}\tilde{R}_i^\dagger$, i.e.
			\begin{equation*}
				\mathbb{P}( \lim\limits_{k\rightarrow \infty}\|\tilde{\Upsilon}_{i,k}^\dagger - \tilde{R}_i^\dagger\|_F = 0) =1.
			\end{equation*}
			
			Thus, under similar arguments of 1), it holds that
			\begin{equation*}
				\mathbb{P}( \lim\limits_{k\rightarrow \infty} \|{P}_{i,k+1|k} - P_i\|_F = 0) =1.
			\end{equation*}
		\end{proof}

		Now we consider the true error of the prior and posterior estimate. Denote the estimate error of the $i$-th node as $\tilde{x}_{i,k|k-1} =x_k -  \hat{x}_{i,k|k-1}$, $\tilde{x}_{i,k|k} =x_k -  \hat{x}_{i,k|k}$, and the corresponding error covariance matrix as 
		\begin{equation*}
			\footnotesize \begin{aligned}
				&\tilde{P}_{i,k|k-1}
				= \mathbb{E}_{\{x_0,w_p, v_{q,r}, p \in \underline{k-1}\cup\{0\},q\in \mathcal{N}, r\in\underline{k-1}\}}[\tilde{x}_{i,k|k-1}(\tilde{x}_{i,k|k-1})^T],\\
				 &\tilde{P}_{i,k|k}
				= \mathbb{E}_{\{x_0,w_p, v_{q,r}, p\in \underline{k-1}\cup\{0\},q\in \mathcal{N}, r\in \underline{k}\}}[\tilde{x}_{i,k|k}(\tilde{x}_{i,k|k})^T].
			\end{aligned}
		\end{equation*}
		
		\begin{theorem}
			\label{theorem:Algorithm1-2-Steady-Error}
			 If $(A, \tilde{C}_i)$ is observable for all $i\in\mathcal{N}$, we have 
			\begin{enumerate}[1)]
			\item  $\tilde{P}_{i,k+1|k}$ in Algorithm~1 converges to $ P_{i}$,
			\item   $ \tilde{P}_{i,k+1|k}$ in Algorithm~2 converges almost surely to  $P_{i}$,
			 \end{enumerate}
			 			where $P_i$ is the solution of \eqref{eq:Algorithm1-2-steady-state-P}.
		\end{theorem}
		
			\begin{proof}
			1) In Algorithm~1, the estimate error satisfies
			\begin{equation}
				\label{eq:theorem2 estimation error}
				\begin{aligned}
					\tilde{x}_{i,k|k}
					& = P_{i,k|k}(P_{i,k|k-1}^{-1}\tilde{x}_{i,k|k-1} -  \tilde{C}_i^T \tilde{U}_{i,k}^{\dagger} \tilde{v}_{i,k}) .
				\end{aligned} 
			\end{equation}
			
			Since $\tilde{x}_{i,k|k-1}$ and $ \tilde{v}_{i,k}$ are independent, the true error covariance matrix follows
			\begin{equation*}
				\begin{aligned}
					\tilde{P}_{i,k+1|k} &= AP_{i,k|k}P_{i,k|k-1}^{-1} \tilde{P}_{i,k|k-1} P_{i,k|k-1}^{-1} P_{i,k|k}A^T \\
					& \phantom{mmm}+  AP_{i,k|k} \tilde{C}_i^T \tilde{U}_{i,k}^{\dagger} \tilde{R}_i \tilde{U}_{i,k}^{\dagger} \tilde{C}_i P_{i,k|k}A^T   + Q. 
				\end{aligned}
			\end{equation*}
			
			By Lemma~\ref{lemma: R dagger}, there exists $\breve{R}_i >0$ such that $\tilde{C}_i^T (\tilde{R}_{i}^{\dagger} - \breve{R}_i^{-1}) \tilde{C}_i = \boldsymbol{O}$.
			
			Denote \begin{equation*}
				\begin{aligned}
					K_{{P}_{i}} & := A {P}_{i} \tilde{C}_i^T \left(\tilde{C}_i{P}_{i}\tilde{C}_i^T+\breve{R}_i\right)^{-1},\\
					\tilde{A}_{{P}_{i} } & := A-  K_{{P}_{i}} \tilde{C}_i,
				\end{aligned}
			\end{equation*}
			 then we have that $ \tilde{A}_{{P}_{i}}$ is Schur stable since $P_i$ is the solution of $DARE(A, \tilde{C}_i, Q, \breve{R}_i)$ and $(A,\tilde{C}_i)$ is observable\cite{kailath2000linear}.
			
			Since  $\tilde{U}_{i,k}^{\dagger}{\rightarrow} \tilde{R}_i^{\dagger}$ by Lemma~\ref{lemma:direct method error}, we have 
			\begin{equation}
				 \tilde{C}_i^T \tilde{U}_{i,k}^{\dagger} \tilde{R}_i \tilde{U}_{i,k}^{\dagger} \tilde{C}_i \rightarrow \tilde{C}_i^T \tilde{R}_{i}^{\dagger} \tilde{C}_i = \tilde{C}_i^T \breve{R}_{i,k}^{-1} \tilde{C}_i.
			\end{equation}
			Consider also $P_{i,k+1|k} {\rightarrow} P_{i}$ by Theorem~\ref{theorem:Algorithm-1-2-P-Convergence}, we have
			\begin{equation*}
				\begin{aligned}
					&AP_{i,k|k}P_{i,k|k-1}^{-1}  \rightarrow \tilde{A}_{P_i},\\
					&AP_{i,k|k} \tilde{C}_i^T \tilde{U}_{i,k}^{\dagger} \tilde{R}_i \tilde{U}_{i,k}^{\dagger} \tilde{C}_i P_{i,k|k}A^T   + Q {\rightarrow} K_{P_i} \breve{R}_i   K_{P_i}^T+ Q.
				\end{aligned}
			\end{equation*}
			
			By exploiting Theorem~1 in \cite{cattivelli2010diffusion}, we have that  $\tilde{P}_{i,k+1|k}$ converges to the solution of the Lyapunov equation
			\begin{equation*}
				X = \tilde{A}_{P_i} X \tilde{A}_{P_i}^T + K_{P_i} \breve{R}_i   K_{P_i}^T+ Q.
			\end{equation*} 
			It can be verified that the above equation is equivalent to \eqref{eq:theorem1 dare eq1}  and consequently equivalent to  \eqref{eq:Algorithm1-2-steady-state-P}.
			
			2) In Algorithm~2, since $P_{i,k+1|k} \stackrel{a.s.}{\longrightarrow} P_{i}$ and $\tilde{\Upsilon}_{i,k}^{\dagger} \stackrel{a.s.}{\longrightarrow} \tilde{R}_i^{\dagger}$, by Lemma~\ref{lemma:stochastic method error}, we have 
			$AP_{i,k|k}P_{i,k|k-1}^{-1} \stackrel{a.s.}{\longrightarrow} \tilde{A}_{P_i}$ and $AP_{i,k|k} \tilde{C}_i^T \tilde{\Upsilon}_{i,k}^{\dagger} \tilde{R}_i \tilde{\Upsilon}_{i,k}^{\dagger} \tilde{C}_i P_{i,k|k}A^T   + Q \stackrel{a.s.}{\longrightarrow} K_{P_i} \breve{R}_i   K_{P_i}^T+ Q $ almost surely.
			
			Similar to the proof of 2) in Theorem~\ref{theorem:Algorithm-1-2-P-Convergence}, by analogous  arguments of 1), $\tilde{P}_{i,k+1|k}$ converges almost surely to $P_i$ in Algorithm~2.
			\end{proof}
		
		\begin{theorem}
			\label{theorem:ModifiedCM best performance}
			The Modified CM converges to the following linear minimum variance estimator  as $k$ tends to infinity, 
			\begin{equation}
				\label{eq:theorem 3 problem}
				\begin{aligned}
				P_{i,k+1|k}^{opt}&:= \min_{K_{i,k}\in \mathbb{R}^{n\times n}}\mathbb{E}\left[(\hat{x}_{i,k+1|k} - x_k)(\hat{x}_{i,k+1|k}-x_k)^T\right]\\
				&\begin{aligned}
				\operatorname{s.t.}\; &\hat{x}_{i,k|k} = (I - K_{i,k}\tilde{C_i})\hat{x}_{i,k|k-1}+ K_{i,k}  \tilde{y}_{i,k},\\
				&\hat{x}_{i,k+1|k} = Ax_{i,k|k}, \;
				\hat{x}_{i,0|0} = \hat{x}_0.
				\end{aligned}
				\end{aligned}				\tag{P1}
			\end{equation}
			
			Furthermore, as the  number of consensus steps per iteration tends to infinity, Algorithms~1-2 converge to the CKF method. 
		\end{theorem}
		
		\begin{proof}
			It is direct that the optimal filter to \eqref{eq:theorem 3 problem} is given by Kalman filter\cite{kailath2000linear}. The corresponding  steady-state error covariance $P_i^{opt}$ satisfies 
			\begin{equation*}
				\begin{aligned}
				P_i^{opt} &= A	P_i^{opt}A^T + Q - AP_i^{opt}\tilde{C}_i^T\\
				&\phantom{mmmmm}\times(\tilde{C}_i	P_i^{opt}\tilde{C}_i^T +\tilde{R}_i)^{\dagger}\tilde{C}_i P_i^{opt}A^T.
				\end{aligned}
			\end{equation*}
			Compared with \eqref{eq:Algorithm1-2-steady-state-P}, it can be verified that $P_{i}^{opt} =P_{i}$. Therefore Algorithms~1-2 converge to the optimal estimator defined by (P1).
			
			As $\gamma$ tends to infinity, we have that $l_{ij}^{(\infty)} = \frac{1}{N}$. Therefore
				$\lim\limits_{\gamma\rightarrow \infty} \tilde{C}_{i}^T \tilde{R}_i^\dagger\tilde{C}_i = \sum_{j\in\mathcal{N}} C_j^T R_j^{-1}C_j$,	$\lim\limits_{\gamma\rightarrow \infty} \tilde{C}_{i}^T \tilde{R}_i^\dagger\tilde{y}_{i,k} = \sum_{j\in\mathcal{N}} C_j^T R_j^{-1}y_{j,k}$.
			Thus the measurement update in  Algorithms~1-2 converges to the measurement update \eqref{eq:sec2 CKF measure update} of CKF with $\gamma \rightarrow \infty$.
		\end{proof}

		\begin{corollary}
		From Theorem~\ref{theorem:ModifiedCM best performance}, it can be inferred that the Modified CM is more accurate than the traditional CM in terms of steady-state performance.
			 Besides, with Corollary 1 in \cite{qian2022consensus}, the steady-state error covariance of the  Modified CM converges to that of the CKF with the speed no slower than the exponential convergence, with $\gamma$ tending to infinity. In other words, there exists $M>0$ and $0<\beta<1$ such that 
		\begin{equation}
			\|P_i-P^{C}\|_2\le M\beta^\gamma, \forall i \in \mathcal{N}, \forall \gamma\ge d,
		\end{equation}
		 where $P_i$ is the steady-state prior estimate error covariance of Algorithms~1-2 (Theorem~\ref{theorem:Algorithm1-2-Steady-Error}), $P^{C}$ is the steady-state error covariance of the  CKF, and $d$ is the diameter of the graph.
		\end{corollary}

		\subsection{Performance Analysis of the Modified CI}
		For brevity, with a slight abuse of notation, we still use $P_i$ to denote the estimated steady-state  covariance of the prior estimate in Algorithms~3-4 (similar to Theorem~\ref{theorem:Algorithm-1-2-P-Convergence}), and use $\tilde{{P}}_{i}$ to denote the actual steady-state error covariance of the prior estimate (similar to Theorem~\ref{theorem:Algorithm1-2-Steady-Error}).
		\begin{theorem}
			\label{theorem:Algorithm-3-4-P-Convergence}
			If A is invertible,  
			\begin{enumerate}[1)]
				\item $P_{i,k+1|k} $ in Algorithm~3 converges to $P_{i}$, 
				\item $P_{i,k+1|k} $ in Algorithm~4 converges almost surely to $P_{i}$, 
			\end{enumerate}
			where $P_i$ is the solution of the following HCRE
			\begin{equation}
				P_{i} = A ( \sum_{j=1}^{N}l_{ij}^{(\gamma)} P_{j}^{-1} + \tilde{C}_i^T \tilde{R}_i^{\dagger} \tilde{C}_i  )^{-1} A^T + Q. \label{eq:Algorithm-3-4-steady-state-P}
			\end{equation}
		\end{theorem}
		\begin{proof}
			In Algorithm~3, we have
			\begin{equation*}
				P_{i,k+1|k} = A ( \sum_{j=1}^{N}l_{ij}^{(\gamma)} P_{j,k|k-1}^{-1} + \tilde{C}_i^T  \tilde{U}_{i,k}^{\dagger}\tilde{C}_i  )^{-1} A^T + Q.
			\end{equation*}
			
			In use of Lemma~\ref{lemma:DARE-HCRE-Convergent-Parameter} and Lemma~\ref{lemma: R dagger}, the proof is completed by taking a similar procedure in the proof in Theorem~\ref{theorem:Algorithm-1-2-P-Convergence}. 
			The statement 2) for Algorithm~4 can be proven similarly.
		\end{proof}
		
		\begin{theorem}
			\label{theorem:Algorithm-3-4-Error}
			If A is invertible, we have 
			\begin{enumerate}[1)]
				\item $\tilde{{P}}_{i,k+1|k} $ in Algorithm~3 converges to $\tilde{{P}}_{i}$, 
				\item $\tilde{{P}}_{i,k+1|k} $ in Algorithm~4 converges almost surely to $\tilde{{P}}_{i}$, 
			\end{enumerate}
			where $\tilde{P}_i$ is the $i$-th diagonal block of $\tilde{\mathcal{P}}$, which is the solution of the following Lyapunov equation
			\begin{equation}
				\mathcal{P} = \mathcal{A} \mathcal{P}  \mathcal{A}^T +\Gamma 
				\diag_{i=1}^N(R_i) \Gamma^T   + \mathbf{1}_N \mathbf{1}_N^T\otimes Q, \label{eq:Algorithm-3-4-steady-error}
			\end{equation}
			where
			\begin{equation*}
				\begin{aligned}
					&\mathcal{A}(i,j) =  l_{ij}^{(\gamma)} A\bar{P}_i P_j^{-1}, \quad 
					 \Gamma(i,j) = l_{ij}^{(\gamma)}A\bar{P}_i\tilde{C}_i^T \tilde{R}_i^{\dagger}C_j^TR_j^{-1},\\
					&\bar{P}_i = (\sum_{j=1}^{N}l_{ij}^{(\gamma)}P_{j}^{-1}+ \tilde{C}_i^T \tilde{R}_i^{\dagger}\tilde{C}_i)^{-1}. 
				\end{aligned}
			\end{equation*}
		\end{theorem}
		
		\begin{proof}
		By Lemma~\ref{lemma: R dagger}, there exists $\breve{R}_i >0$ such that $\tilde{C}_i^T (\tilde{R}_{i}^{\dagger} - \breve{R}_i^{-1}) \tilde{C}_i = \boldsymbol{O}$, which follows 
		\begin{equation}
			\bar{P}_i = (\sum_{j=1}^{N}l_{ij}^{(\gamma)}P_{j}^{-1}+ \tilde{C}_i^T \breve{R}_i^{-1}\tilde{C}_i)^{-1},
		\end{equation}
		and \eqref{eq:Algorithm-3-4-steady-state-P} can be rewritten as
		\begin{equation}
			P_{i} = A ( \sum_{j=1}^{N}l_{ij}^{(\gamma)} P_{j}^{-1} + \tilde{C}_i^T \breve{R}_i^{-1} \tilde{C}_i  )^{-1} A^T + Q.
		\end{equation}

		By the same arguments in Section IV.A in \cite{qian2022harmonic}, it can be obtained that $ \mathcal{A}$ is Schur stable.
		
		Denote $\mathcal{A}_k(i,j) =  l_{ij}^{(\gamma)} A{P}_{i,k|k} P_{j,k|k-1}^{-1}$, $ \Gamma_k (i,j) = l_{ij}^{(\gamma)}A{P}_{i,k|k}\tilde{C}_i^T \tilde{R}_i^{\dagger}C_j^TR_j^{-1}$.
		
		1) In Algorithm~3, similar to \eqref{eq:theorem2 estimation error}, we can derive that
		\begin{equation*}
			\small \begin{aligned}
				\tilde{x}_{i, k|k} &=x_k -  \hat{x}_{i,k|k}\\
				& = P_{i,k|k}(\sum_{j=1}^{N}l_{ij}^{(\gamma)} ( P_{j,k|k-1}^{-1} \tilde{x}_{j,k|k-1} -  \tilde{C}_i^T \tilde{U}_{i,k}^{\dagger} C_j R_j^{-1}{v}_{j,k}) ).
			\end{aligned} 
		\end{equation*}
		
		By denoting $ \tilde{x}_{k+1|k} = \col_{i=1}^N(\tilde{x}_{i, k+1|k})$, we have
		\begin{equation*}
			\tilde{x}_{k+1|k}  = \mathcal{A}_k  \tilde{x}_{k|k-1} + \Gamma_k \col_{j=1}^N(v_{j,k}) + \mathbf{1}_N \otimes w_{k}.
		\end{equation*}
		
		Thus
		\begin{equation*}
			\begin{aligned}
				\tilde{\mathcal{P}}_{k+1|k}
				= &\mathbb{E}_{\{w_p, \,v_{q,r},\, p \in \underline{k}\cup\{0\},\,q\in \mathcal{N}, r\in\underline{k}\}}[\tilde{x}_{k+1|k}(\tilde{x}_{k+1|k})^T]\\
				=&  \mathcal{A}_{k} \tilde{\mathcal{P}}_{k|k-1} \mathcal{A}_{k}^T +  \Gamma_k \diag_{i=1}^N(R_i)  \Gamma_k +  \mathbf{1}_N \mathbf{1}_N^T\otimes Q.
			\end{aligned}
		\end{equation*}
		
		Since $P_{i,k+1|k} {\longrightarrow} P_{i}$ and $\tilde{U}_{i,k}^{\dagger}{\rightarrow} \tilde{R}_i^{\dagger}$, we have $\mathcal{A}_{k}\longrightarrow \mathcal{A}$ and $\Gamma_k \longrightarrow \Gamma$. By exploiting Theorem~1 in \cite{cattivelli2010diffusion}, $\mathcal{P}_{k+1|k}$ converges to the solution of \eqref{eq:Algorithm-3-4-steady-error}.
		
		2) In Algorithm~4, since $P_{i,k+1|k} \stackrel{a.s.}{\longrightarrow} P_{i}$ and $\tilde{\Upsilon}_{i,k}^{\dagger} \stackrel{a.s.}{\longrightarrow} \tilde{R}_i^{\dagger}$, by Lemma~\ref{lemma: a.s. mat prod mat inv}, we have $\mathcal{A}_{k}\stackrel{a.s.}{\longrightarrow} \mathcal{A}$ and $\Gamma_k \stackrel{a.s.}{\longrightarrow }\Gamma$. By exploiting Theorem~1 in \cite{cattivelli2010diffusion}, $\mathcal{P}_{k+1|k}$ converges almost surely to the solution of \eqref{eq:Algorithm-3-4-steady-error}.
		\end{proof}

		\begin{theorem}
			\label{theorem:MCI consistency}
			When steady state is achieved, the pair $(\hat{x}_{i,k|k}, P_{i,k|k})$ and  $(\hat{x}_{i,k+1|k}, P_{i,k+1|k})$ in Algorithms~3-4 are consistent. 
		\end{theorem}
		
		\begin{proof}
			Consider the steady-state version of Algorithms~3-4 by merely substituting  $\tilde{U}_{i,k}$ or $\tilde{\Upsilon}_{i,k}$ with the static quantity $\tilde{R}_i$, and we call this algorithm the steady Modified CI.  It can be directly obtained that in this algorithm, $P_{i,k+1|k}\rightarrow P_i$ in Theorem~\ref{theorem:Algorithm-3-4-P-Convergence} and $\tilde{P}_{i,k+1|k} \rightarrow \tilde{P}_{i}$ in Theorem~\ref{theorem:Algorithm-3-4-Error}.
			
			 In the steady Modified CI, we have 
			\begin{equation*}
				\small\begin{aligned}
						P_{i,k|k} &= (\sum_{j=1}^{N}l_{ij}^{(\gamma)} P_{j,k|k-1}^{-1}+\tilde{C}_i^T  \tilde{R}_{i}^{\dagger} \tilde{C}_i)^{-1},\\
						\tilde{x}_{i,k|k} &=  P_{i,k|k}(\sum_{j=1}^{N}l_{ij}^{(\gamma)} ( P_{j,k|k-1}^{-1} \tilde{x}_{j,k|k-1} -  \tilde{C}_i^T \tilde{R}_{i}^{\dagger} C_j R_j^{-1}{v}_{j,k}) ).
				\end{aligned}
			\end{equation*}
			
			If $\mathbb{E}[\tilde{x}_{i,k|k-1}\tilde{x}_{i,k|k-1}^T]\le P_{i,k|k-1}$ holds for all $i\in \mathcal{N}$, 
			from covariance intersection\cite{julier1997non}, we have
			\begin{equation*}
				\begin{aligned}
				&\mathbb{E}\left[(\sum_{j=1}^{N}l_{ij}^{(\gamma)} ( P_{j,k|k-1}^{-1} \tilde{x}_{j,k|k-1})(\sum_{j=1}^{N}l_{ij}^{(\gamma)} ( P_{j,k|k-1}^{-1} \tilde{x}_{j,k|k-1})^T\right] \\
				\le& \sum_{j=1}^{N}l_{ij}^{(\gamma)} P_{j,k|k-1}^{-1}.
				\end{aligned}
			\end{equation*}
			
			On the other hand, from \eqref{eq:covariance of  fuesed measurement}, we have 
			\begin{equation*}
				\begin{aligned}
				&\mathbb{E}\left[(\sum_{j=1}^{N}l_{ij}^{(\gamma)} \tilde{C}_i^T \tilde{R}_{i}^{\dagger} C_j R_j^{-1}{v}_{j,k})(\sum_{j=1}^{N}l_{ij}^{(\gamma)} \tilde{C}_i^T \tilde{R}_{i}^{\dagger} C_j R_j^{-1}{v}_{j,k})^T\right] \\
				= &\tilde{C}_i^T  \tilde{R}_{i}^{\dagger} \tilde{C}_i.
				\end{aligned}
			\end{equation*}
		
			Since $\tilde{x}_{i,k|k-1}$ is independent of $v_{j,k}$, for all $i,j \in \mathcal{N}$, we have
			\begin{equation*}
				\begin{aligned}
				\mathbb{E}\left[	\tilde{x}_{i,k|k} 	\tilde{x}_{i,k|k} ^T\right] &\le P_{i,k|k}(\sum_{j=1}^{N}l_{ij}^{(\gamma)} P_{j,k|k-1}^{-1} +\tilde{C}_i^T  \tilde{R}_{i}^{\dagger} \tilde{C}_i ) P_{i,k|k}\\
				& = P_{i,k|k}.
				\end{aligned}
			\end{equation*}
			and therefore $ \mathbb{E}[	\tilde{x}_{i,k+1|k} 	\tilde{x}_{i,k+1|k} ^T] \le P_{i,k+1|k}$ by the prediction step.
			In combination with the initial condition $ 	\mathbb{E}[	\tilde{x}_{i,0|0} 	\tilde{x}_{i,0|0} ^T] = P_{i,0|0}$, we have that the pair $(\hat{x}_{i,k|k}, P_{i,k|k})$ and  $(\hat{x}_{i,k+1|k}, P_{i,k+1|k})$ are consistent in the steady Modified CI. Note that the steady Modified CI and Algorithms~3-4 share the same actual and estimated error covariance in steady state.
			Thus, in Algorithms~3-4, it still holds that $\tilde{P}_{i}\le P_i $.

		\end{proof}
		
		\begin{theorem}
			As the number of consensus steps per iteration tends to infinity, Algorithms~3-4 converge to the CKF method. 
		\end{theorem}
		\begin{proof}
			Note that $\lim\limits_{\gamma\rightarrow \infty}l_{ij}^{(\gamma)}=\frac{1}{N}$, thus $P_{i,k+1|k} = P_{j,k+1|k}$ for all $i,j \in \mathcal{N}, k\in \mathbb{Z}_0^+$. The remaining proof is similar to that of Theorem~\ref{theorem:ModifiedCM best performance}.
		\end{proof}

	\begin{remark}
	\label{remark:smaller bound}
		The traditional CI also provides an upper bound of the actual error covariance (see Theorem~3 in \cite{battistelli2014kullback}). Compared with that, the Modified CI provides a smaller bound.
	This can be justified by comparing the HCRRs that the covariance bound \( P_i \) must satisfy. For the Modified CI, the bound $P_i^{MCI}$ corresponds to the solution of \eqref{eq:Algorithm-3-4-steady-state-P}, while for CI, the bound $P_i^{CI}$ is obtained by replacing  $\tilde{C}_i^T\tilde{R}_i^{\dagger}\tilde{C}_i$ in \eqref{eq:Algorithm-3-4-steady-state-P} with \( \tilde{C}_i \). Note that \( \tilde{C}_i \le \tilde{C}_i^T\tilde{R}_i^{\dagger}\tilde{C}_i \),  therefore $P_i^{MCI} \le P_i^{CI}$.
	
	An evident difference is that as the  number of consensus steps per iteration approaches infinity, the gap between the estimated and true error covariance in CI converges to a constant non-zero matrix. In contrast, in the Modified CI, this gap converges to zero, which means that both of the estimated and true error covariance converge to the error covariance of CKF.
	\end{remark}
	
	\begin{remark}
	The proposed methods, along with their stability and performance analysis (excluding only the upper bound of the convergence rate in Lemma~\ref{lemma:direct method error}, \(\frac{N|\lambda_2|^k}{1 - N|\lambda_2|^{k}}\)), can be readily extended to strongly connected directed graphs, provided that the weight matrix is doubly stochastic with positive diagonal entries. This type of weight matrix can be computed in a distributed manner following \cite{gharesifard2012distributed}, and the proof follows directly using Lemma~1 from \cite{qian2021fully}.
	\end{remark}

	\section{Simulation}
		This section aims at validating the theoretical analysis given in Sections III-IV and showing the effectiveness of the proposed methods. Three numerical experiments are provided for a distributed target-tracking problem. In the first experiment, the proposed algorithms are implemented on a randomly generated network, and the simulated performance is compared with the theoretically predicted performance derived in Theorems~\ref{theorem:Algorithm1-2-Steady-Error} and Theorem~\ref{theorem:Algorithm-3-4-Error}. Meanwhile, the  error of estimating $\tilde{R}_i$ and $\tilde{R}_i^\dagger$ in this scenario is also compared with the bounds given in Lemmas~\ref{lemma:direct method error} and Lemma~\ref{lemma:stochastic method error}. In the second and third experiments, we consider different communicating topologies and different consensus weights, respectively, and compare the performance of Algorithms~1-4 with traditional methods.
		
		The considered target tracking problem can be found in \cite{battistelli2014kullback}. The  system is  described by $x_{k+1} = Ax_k + w_k$, where $x_k = [p_x, v_x, p_y, v_y]^T$ denotes the target position and velocity along the axis, the system matrix $A$ is given by
		\begin{equation*}
			a = \begin{pmatrix}
				1&T\\0 &1
			\end{pmatrix},\quad
			A = \begin{pmatrix}a & \boldsymbol{0}_{2\times 2}\\
				\boldsymbol{0}_{2\times 2} & a
			\end{pmatrix},
		\end{equation*}
	and the error covariance $Q$ is given by
		\begin{equation*}
			G=\left(\begin{array}{cc}
				\frac{T^3}{3} & \frac{T^2}{2} \\
				\frac{T^2}{2} & T
			\end{array}\right), \quad Q=\left(\begin{array}{cc}
				G & 0.5 G \\
				0.5 G & G
			\end{array}\right)
		\end{equation*}
		with the sampling interval $T = 0.1s$.
		
		We assume that the network contains three kinds of nodes, each of which can process local data and communicate with its neighbours. The difference lies in the measurement it can obtain. Type 1 and  Type 2 can respectively obtain the $x$, $y$ position of the target, while Type 3 does not have the ability of sensing. The observation matrices are	\begin{equation*}
			\begin{aligned}
				C^{(1)} = [1, 0,0,0],\quad C^{(2)} = [0,0,1,0], \quad C^{(3)} = \boldsymbol{0}_{1\times4},
			\end{aligned}
		\end{equation*}
		and the noise covariances are
		\begin{equation*}
			R^{(1)} = R^{(2)} = 0.01, \quad R^{(3)} = 10^6.
		\end{equation*}
		
		In the first experiment, the whole network is composed of $N=20$ nodes, randomly spreading across a square region with side length 300m. We assume that any two nodes within a distance of 100m from each other can communicate with each other, and the consensus weights are set to be Metropolis weights\cite{xiao2005scheme}. The communication topology generated in this case is shown in Fig. \ref{fig:fig1}. 
		
		\begin{figure}[t]
			\centering
			\includegraphics[width=0.9\linewidth]{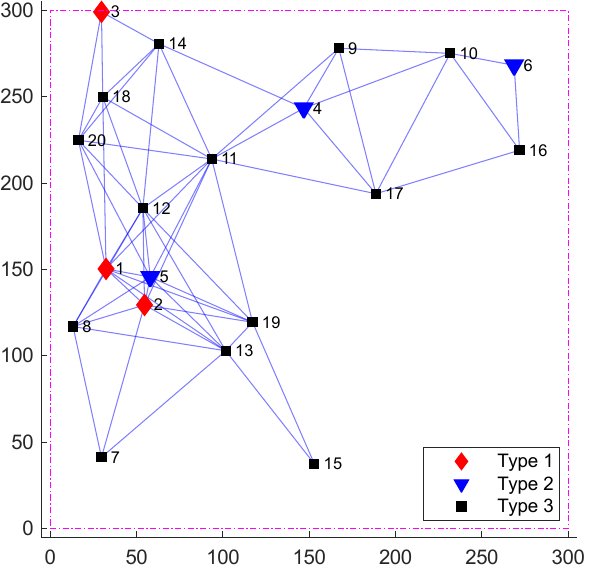}
			\caption{Plot of the randomly generated sensor network (20 nodes and communication radius equal to 100m).}
			\label{fig:fig1}
		\end{figure}
		
		We consider a trajectory of length 20s, starting from $x_0 \sim \mathbf{N}(\hat{x}_0, 100I_4)$ where $\hat{x}_0 = [150, 0, 150, 0]^T$. With Monte Carlo method, each filtering process is run for 1000 times and the mean square error (MSE) is calculated as 
		
		\begin{equation*}
			\mathrm{MSE}_{i, k}=\frac{1}{1000} \sum_{l=1}^{1000}\left\|\hat{x}_{i, k|k}^{(l)}-x_k^{(l)}\right\|_2^2,
		\end{equation*}
		where $\hat{x}_{i, k|k}^{(l)}$ and $x_k^{(l)}$ denote the posterior state estimate and true state at
		time instant $k$ in the $l$-th simulation, respectively. In the following discussion, we use the mean of  MSE (MMSE)  in steady state of the whole network to describe the algorithm performance.
		
		\begin{figure}
			\centering
			\includegraphics[width=0.9\linewidth]{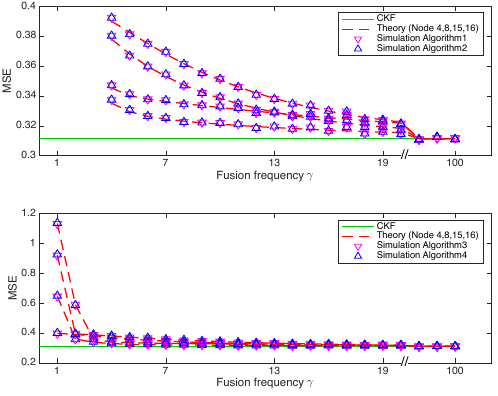}
			\caption{The steady-state errors for Algorithms~1-2 (top) and Algorithms~3-4 (bottom), along with theoretical predictions from Theorem~2 (top) and Theorem~5 (bottom).}
			\label{fig:fig2}
		\end{figure}
		The simulated  and theoretically predicted performance of Algorithms~1-4 under different $\gamma$ (representing the number of communication steps between two successive sampling instants) is shown in Fig. \ref{fig:fig2}. It can be seen the Monte Carlo result perfectly fits the theoretical analysis, and both of the steady-state performance of  the Modified CM and Modified CI converge to the optimal performance of CKF as expected.
		
		\begin{figure}
			\centering
			\includegraphics[width=0.9\linewidth]{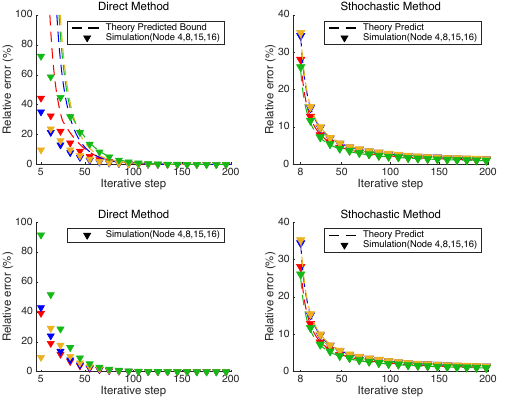}
			\caption{The estimation errors for QWS (top) and the corresponding Moore-Penrose inverse (bottom) using both direct and stochastic methods. The theoretical bounds/values are based on Lemma~1 (left) and Lemma~3 (right).}
			\label{fig:fig3}
		\end{figure}
		Fig. \ref{fig:fig3} shows the error of the direct and stochastic method on estimating $\tilde{R}_i$ and $\tilde{R}_i^\dagger$ in this experiment,  where $\gamma = 5$. 	The direct method is used in Algorithm~1 and Algorithm~3 and the stochastic method in Algorithm~2 and Algorithm~4. It can bee seen that in both  methods, the estimate errors converge to zero rapidly, and both of the predicted error bound in Lemma~\ref{lemma:direct method error} and the predicted error in Lemma~\ref{lemma:stochastic method error} precisely characterize the decay of the error.

		\begin{figure}[ht]
			\centering
			\includegraphics[width=0.9\linewidth]{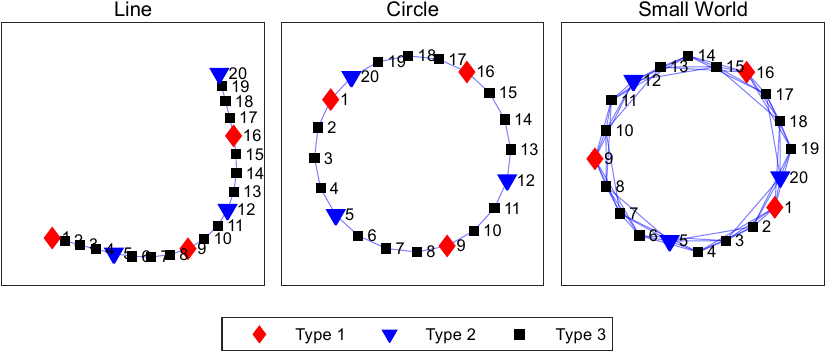}
			\caption{Plot of sensor networks added in Experiment 2.}
			\label{fig:fig4}
		\end{figure}
		
				\begin{figure}
			\centering
			\includegraphics[width=0.9\linewidth]{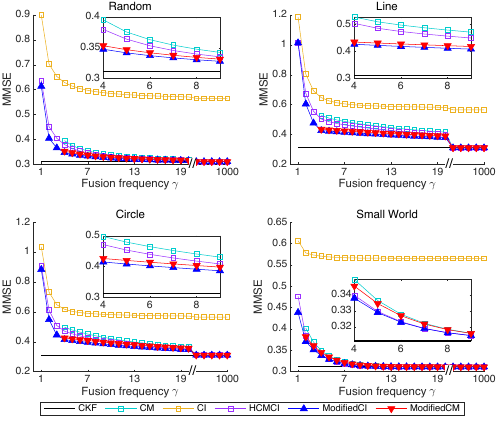}
			\caption{The steady-state performance of the proposed and conventional algorithms under different communication topologies in Experiment 2.}
			\label{fig:fig5}
		\end{figure}

		Experiment 2 considers four communication graphs, i.e., randomly generated, line, circle and small world graphs. The newly added graphs are shown in Fig. \ref{fig:fig4}. We implement Algorithms~1-4 and the CM\cite{qian2022consensus}, CI\cite{battistelli2014kullback}, HCMCIs\cite{battistelli2014consensus} on these graphs, while the consensus weights are set to be Metropolis weights as well.  The result for different  number of consensus steps per iteration is shown in Fig. \ref{fig:fig5}. It can be seen that the performance of the Modified CM is always better than that of CM as shown in Theorem~\ref{theorem:ModifiedCM best performance}, and the Modified CI outperforms all the other algorithms on each graph in this experiment. The reason for the performance improvement can be attributed to the use of accurate fused measurement covariance in the modified algorithm, rather than the approximated one used in previous algorithms, as discussed in Section II, Theorem~\ref{theorem:ModifiedCM best performance} and Remark~\ref{remark:smaller bound}.

		\begin{table}[ht]
			\centering
			\caption{Comparison of the steady-state performance of different algorithms in experiment 3.}
			\label{tab:experiment3}
			\fontsize{7}{10}\selectfont
			\begin{tabular}{|@{\hspace{3pt}}c@{\hspace{3pt}}| *{5}{c@{\hspace{3pt}}c@{\hspace{3pt}}|}}
				\hline
				{} & \multicolumn{2}{c|}{CM} & \multicolumn{2}{c|}{CI} & \multicolumn{2}{c|}{HCMCI} & \multicolumn{2}{c|}{Modified CM} & \multicolumn{2}{c|}{Modified CI} \\
				\hline
				$\eta$ &abs & rel(\%) &abs & rel(\%) &abs & rel(\%) &abs & rel(\%) &abs & rel(\%) \\
				\hline
				0& 0.395 & 100.0 &0.629 & 100.0 &0.380 & 100.0 &0.353 & 100.0 &\textbf{0.348} & 100.0 \\
				\hline
				0.1& 0.406 & 102.8 &0.637 & 101.3 &0.388 & 102.2 & 0.357 & 100.9 &\textbf{0.351} & \textbf{100.9} \\
				\hline
				0.3& 0.437 & 110.6 &0.659 & 104.8 &0.410 & 107.9  &0.364 & 103.1 &\textbf{0.358} &  \textbf{103.0}\\
				\hline
				0.5& 0.491 & 124.2 &0.700 &111.4  &0.443 & 116.7 &0.374 & 105.9 &\textbf{0.367}& \textbf{105.6} \\
				\hline
				0.7&0.626 & 158.2 &0.797 & 126.7 & 0.503 & 132.4 &0.385 & 108.9 &\textbf{0.378} & \textbf{108.7} \\
				\hline
				0.9&2.480 & 627.3 &1.247 & 198.3  &0.706 & 186.0 &0.395 & \textbf{111.7} &\textbf{0.390} & 112.2 \\
				\hline
			\end{tabular}
		\end{table}
		
		In the third experiment, we compare the performance of the previous and modified algorithms on the graph generated in Experiment 1, with the consensus weights changing from 
		 L(Metropolis weights) to $\eta I+(1-\eta)L$, where $0 \le \eta <1$. This change corresponds to the scenario that each node chooses to place more confidence on its own measurement and pay less attention to others'. It can be directly verified that the adjusted weight matrices are still doubly stochastic. The result with $\gamma=4$ is shown in Table \ref{tab:experiment3}, where the ``abs'' column shows the absolute value of MMSE, and the ``rel" column shows the ratio of MMSE under specific weights to MMSE under the Metropolis weights. It can be observed that the performance of all the algorithms is deteriorated as $\eta$ increases, but the degradation of the modified methods is significantly smaller than the traditional methods. This can still be attributed to that the traditional methods use an approximated covariance of the fused measurement, and the approximating accuracy is dependent on the the consensus weights. While in the modified methods, this problem does not exist since the accurate covariance is used. This experiment shows that the modified methods are relatively less sensitive to the change of weights.

		\section{Conclusion}
		In this paper, we demonstrate that the accurate fused measurement noise covariance is crucial for consensus-based DKFs, as its absence can lead to performance degradation or inconsistency. We derive two methods, the direct and the stochastic approaches, to compute this covariance in a fully distributed manner, which are then used to modify the CM and CI methods. Through an analysis of DARR/HCRR with asymptotically convergent parameters, we provide both stability guarantees and the steady-state covariance of the estimate. This analysis enriches the DARE/HCRE theory and offers practical techniques for designing DKF schemes with convergent parameters. The modified algorithms are shown to outperform traditional CM/CI algorithms in terms of steady-state estimation accuracy, resolve the inherent inconsistency and non-optimality in CM, and bridge the insurmountable performance gap between CI and CKF. Overall, this paper presents a paradigm for addressing flawed measurement fusion, which is promising to be adapted to other consensus-based DKFs.
		
		However, covariance intersection still induces conservatism in the Modified CI, as accurately characterizing the complex correlations between estimates from different nodes remains challenging. This issue is significant and is worth further investigation. Additionally, the design of optimal weights to minimize the covariance of fused measurements is also an important area for future research.

		\section*{Appendix}		
		\subsection{Proof of Lemma~\ref{lemma:direct method error}}
		\begin{lemma}
			\label{lemma:direct method core}
			Assume that \rm{$\mathbi{Q}_p$} $= \col_{i=1}^p(q_i) \in \mathbb{R}^{p\times N}$ is full row rank, where $q_i \in \mathbb{R}^{1\times N}$. Let $\alpha_{i} > 0$, then
			\begin{equation}
				\label{eq:direct-method-core}
			\textstyle	\alpha_{i} q_i \left(\sum_{j=1}^{p}\alpha_{j} q_j^T q_j\right)^\dagger q_k^T = \delta_{ik}. 
			\end{equation}
		\end{lemma}
		\begin{proof}
			Denote $\bar{Q} = \col_{i=1}^p(\sqrt{\alpha_{i}} q_i) $, with the definition of Moore-Penrose inverse, we have
			\begin{equation}
				\label{eq:direct-method-proof-eq1}
				\bar{Q}^{T} \bar{Q} \left(\bar{Q}^{T}\bar{Q}\right)^\dagger \bar{Q}^T  \bar{Q}=\bar{Q}^T  \bar{Q} .
			\end{equation}
			Since $\mathbi{Q}_p$ is full row rank, we have $\bar{Q}$ is full row rank, and thus $\bar{Q}\bar{Q}^T$ is invertible. Multiplying both sides of the equation \eqref{eq:direct-method-proof-eq1} with $(\bar{Q}\bar{Q}^T)^{-1}\bar{Q}$  on the left and $\bar{Q}^T(\bar{Q}\bar{Q}^T)^{-1}$ on the right, one can show
			\begin{equation*}
				\bar{Q} \left(\bar{Q}^{T}\bar{Q}\right)^\dagger \bar{Q}^T = I_{p},
			\end{equation*}
			which is equivalent to \eqref{eq:direct-method-core}.
		\end{proof}

		\begin{proof}[\textbf{Proof of Lemma~\ref{lemma:direct method error}}]
			The proof of 2) will be given firstly.
				Since $\mathbi{Q} = \col_{i=1}^N(q_i)$ is full rank, we have $\col_{i\in \mathcal{J}_i^{(k)}}(q_i)$ is  full row rank, where $\mathcal{J}_i^{(k)} = \{j\in \mathcal{N}| l_{ij}^{(k)} >0\}$ represents the neighbours within $k$ steps from the $i$-th node.  Denote $\tilde{q}_i := q_i\otimes I_n$, then by Lemma~\ref{lemma:direct method core} we have 
			\begin{equation}
				\label{eq:proof of lemma1 eq 1}
				\alpha_{i} \tilde{q}_i \left(\sum_{j\in \mathcal{J}_{i}^{(k)}}\alpha_{j} \tilde{q}_j^T \tilde{q}_j\right)^\dagger \tilde{q}_s^T = \delta_{is}I_n, 
			\end{equation} for any $i,s \in \mathcal{J}_{i}^{(k)} $, $\alpha_{i}> 0$ and $k \ge \gamma$.
			
			 In use of \eqref{eq:proof of lemma1 eq 1} and $ \mathcal{J}_{i}^{(\gamma)} \subseteq   \mathcal{J}_{i}^{(k)}$,  we get
			\begin{equation}
				\label{eq:direct-method-Uik}
				\begin{aligned}
					\tilde{U}_{i,k} 
				=&(\sum_{t=1}^N l_{it}^{(\gamma)}Y_t^T \tilde{q}_t ) ( \sum_{j=1}^{N}l_{ij}^{(k)}N \tilde{q}_j^T\tilde{q}_j)^\dagger   (\sum_{s=1}^N l_{is}^{(\gamma)}Y_s^T \tilde{q}_s  )^T\\
				= &\frac{1}{N}\sum_{t\in \mathcal{J}_i^{(\gamma)}}\sum_{s\in \mathcal{J}_i^{(\gamma)} } l_{it}^{(\gamma)} l_{is}^{(\gamma)} Y_t^T \tilde{q}_t ( \sum_{j\in \mathcal{J}_i^{(k)} }l_{ij}^{(k)}\tilde{q}_j^T\tilde{q}_j)^\dagger \tilde{q}_s^T Y_s \\
				=& \frac{1}{N}\sum_{j=1}^N \frac{(l_{ij}^{(\gamma)})^2}{l_{ij}^{(k)}}X_j.
			\end{aligned}
			\end{equation}
			
			Therefore, 
			\begin{equation}
				\label{eq:direct-method-error-inequality}
				\begin{aligned}
					\|\tilde{U}_{i,k} - \tilde{X}_i\|_2 
				=& \left\|\sum_{j=1}^N (\frac{1}{N l_{ij}^{(k)}}-1) (l_{ij}^{(\gamma)})^2X_j  \right\|_2\\
				\le & \max_{j\in \mathcal{J}_i^{(\gamma)}}  |\frac{1}{Nl_{ij}^{(k)}}-1| \|\tilde{X}_i\|_2,
			\end{aligned}
			\end{equation}
			where the inequality is due to  $l_{ij}^{(k)}>0$ for $j \in  \mathcal{J}_i^{(k)}$ and the fact that $\|D\|_2 \ge \|E\|_2$ holds for any symmetric matrices $ D, E$ satisfying $D \ge E \ge 0 $.

			On the other hand, since $\mathcal{L}$ is doubly stochastic and symmetric, it can be decomposed by
				$\mathcal{L} = U^{T} \Sigma U$, 
			where $U$ is an orthogonal matrix with the first row $\frac{1}{N} \mathbf{1}_N$, and $\Sigma = \diag_{i=1}^N(\lambda_i)$ where $\lambda_i \in \mathbb{R}$ is the $i$-th eigenvalue of $\mathcal{L}$ satisfying $1 = \lambda_1 \ge |\lambda_2| \ge \dots \ge |\lambda_N|$.
			Besides, since the graph is connected, $1$ is a simple eigenvalue of $\mathcal{L}$, and  $|\lambda_2|<1$.

			Thus we have $\| \mathcal{L}^{k}  - \frac{1}{N}\mathbf{1}_N \mathbf{1}_N^T\|_2 = |\lambda_2|^{k}$,
			which leads to
			\begin{equation}
				\textstyle	\label{eq:proof of lemma1 eq2}
				|l_{ij}^{(k)} - \frac{1}{N}| \le |\lambda_2|^k, \quad \forall i, j\in \mathcal{N}.
			\end{equation}
			
			Since $|\lambda_2| <1$, there exists $k_0\ge \gamma$ satisfying $|\lambda_2|^{k_0} < 1/N$. When $k \ge k_0$, from \eqref{eq:proof of lemma1 eq2} we get
			\begin{equation}
			\textstyle	\label{eq:proof of lemma1 eq4}
				l_{ij}^{(k)} \ge \frac{1}{N} - |\lambda_2|^k \ge  \frac{1}{N} - |\lambda_2|^{k _0} >0, \forall i, j \in \mathcal{N},
			\end{equation}
			thus
			\begin{equation}
				\textstyle \label{eq:proof of lemma1 eq3}
				|\frac{1}{Nl_{ij}^{(k)}}-1|  \le \frac{N|\lambda_2|^k}{1-N|\lambda_2|^k}.
			\end{equation}

			Substituting \eqref{eq:proof of lemma1 eq3} into \eqref{eq:direct-method-error-inequality}, the proof of 2) is completed.

			Now we are ready to prove 1). Utilizing the inequalities in \eqref{eq:direct-method-error-inequality}\eqref{eq:proof of lemma1 eq3}, it is straightforward to show that $\lim\limits_{k\rightarrow \infty}\|\tilde{U}_{i,k}-\tilde{X}_{i}\|_2=0$. 
			
			Because $\tilde{X}_i = \tilde{X}_i^T >0$ is symmetric, it has the following decomposition 
			\begin{equation}
				\label{eq:proof of Lemma1 eq6}
				\tilde{X}_i = V_i^T \diag(	\Sigma_{i}, \boldsymbol{O}) V_i,
			\end{equation}
			where $V_i$ is an orthogonal matrix and  $0<{\Sigma}_i \in \mathbb{R}^{r_i\times r_i}$ is a diagonal matrix.

			Since  $\tilde{X}_i =\sum_{j=1}^{N}(l_{ij}^{(\gamma)})^2 Y_j^TY_j $, any matrix $V_{i,\perp}$ satisfies $V_{i,\perp} \tilde{X}_{i} V_{i,\perp}^T = \boldsymbol{O}$ if and only if 
			\begin{equation}
				\label{eq:proof of Lemma1 eq5}
				Y_jV_{i,\perp}^T =\boldsymbol{O}, \forall j \in \mathcal{J}_{i}^{(\gamma)}.
			\end{equation}

			On the other hand, it follows from  \eqref{eq:proof of lemma1 eq4} that $l_{ij}^{(k)}>0$, $\forall k>k_0$, $\forall i,j \in \mathcal{N}$. Together with
			$	\tilde{U}_{i,k}  = \frac{1}{N}\sum_{j=1}^N \frac{(l_{ij}^{(\gamma)})^2}{l_{ij}^{(k)}}X_j$, we can derive that any matrix $V_{i,\perp}$ satisfies $V_{i,\perp} \tilde{U}_{i,k} V_{i,\perp}^T = \boldsymbol{O}$ if and only if \eqref{eq:proof of Lemma1 eq5} holds, when $k> k_0$.
			
			By \eqref{eq:proof of Lemma1 eq6}, $
				V_i \tilde{X}_{i} V_i^T = \diag(	\Sigma_{i}, \boldsymbol{O})$. Therefore, 
			\begin{equation*}
				V_i \tilde{U}_{i,k} V_i^T = \diag(	\Sigma_{i,k}, \boldsymbol{O}),
			\end{equation*}
			where  $0<{\Sigma}_{i,k} \in \mathbb{R}^{r_i\times r_i}$ .
			
			Since $\tilde{U}_{i,k}\rightarrow\tilde{X}_i$, it follows that $\Sigma_{i,k} \rightarrow \Sigma_{i}$, and  $\Sigma_{i,k}^{-1} \rightarrow \Sigma_{i}^{-1}$ \cite{stewart1969continuity}. Thus 
			\begin{equation*}
				\tilde{U}_{i,k}^\dagger = V_i^T \diag(	\Sigma_{i,k}^{-1}, \boldsymbol{O})  V_i\rightarrow   V_i^T  \diag(	\Sigma_{i}^{-1}, \boldsymbol{O}) V_i = \tilde{X}_i^\dagger.
			\end{equation*}  	
		\end{proof}	
		
		\subsection{Proof of Lemmas~\ref{lemma: a.s. mat prod mat inv}-\ref{lemma:stochastic method error}}
		In order to prove Lemma~\ref{lemma: a.s. mat prod mat inv} and Lemma~\ref{lemma:stochastic method error},
		the following definition and two lemmas are required, which can be found in Chapter 2.1  in \cite{van2000asymptotic}.

		\begin{lemma}[Strong Law of Large Number]
			\label{lemma:SLLN}
			Let $X_1, X_2, \dots $ be pairwise independent identically distributed random vectors with $E[|X_i|]<\infty$. Let $ E[X_i] = \mu$ and $S_k = X_1+\dots + X_k$. Then $S_k/k \rightarrow \mu$ almost surely as $k \rightarrow \infty$.
		\end{lemma}
		
		\begin{lemma}[Continuous Mapping]
			\label{lemma:AlmostSureContinuous}
			Let $g:\mathbb{R}^n \rightarrow \mathbb{R}^m$ be continuous at $X$. If $X_k \stackrel{a.s.}{\longrightarrow}$ X, then $g(X_k) \stackrel{a.s.}{\longrightarrow} g(X)$.
		\end{lemma}
		
		To apply the strong law of large numbers on the sequence of  auxiliary variables $\tilde{U}_{i,k}$, a slight modification is needed, as stated in the following lemma.
		
		\begin{lemma}
			\label{lemma:SLLN for matrix}
			Let the sequence of pairwise independent random variables $\{u_k\}_{k=1}^\infty$ satisfy $u_k\sim \mathbf{N}_d(0,\Sigma)$, then ${\sum_{k = 1}^{n}u_k (u_k)^T}/{n} $ converges almost surely to $\Sigma$.

		\end{lemma}
		\begin{proof}
			It only needs to prove that each element of the random matrix $ u_k (u_k)^T$ converges almost surely to the corresponding element in $\Sigma$. 
			
			Since $u_k\sim \mathbf{N}_d(0,\Sigma)$, we have $\mathbb{E}[u_k (u_k)^T] = \Sigma$. Furthermore, denote $u_k(i)$ as the $i$-th entry of $u_k$, the sequence $\{u_k(i) u_k(j)\}_{k=1}^\infty $ is pairwise independent identically distributed.  Together with 
			\begin{equation*}
					\mathbb{E}[\left|u_k(i) u_k(j) \right|] \le \mathbb{E}[(u_k(i))^2+ (u_k(j))^2] \le \operatorname{Tr}\{\Sigma\} < \infty.
			\end{equation*}
			and Lemma~\ref{lemma:SLLN}, the proof is completed. 
		\end{proof}

		\begin{proof}[\textbf{Proof of Lemma~\ref{lemma: a.s. mat prod mat inv}}]
			\begin{enumerate}[1)]
				\item From Definition 1, we have 
				\begin{equation*}
					\small \begin{aligned}
						 \mathbb{P}(\lim_{k\rightarrow \infty} \|D_k - D\|_F = 0)=
						  \mathbb{P}(\lim_{k\rightarrow \infty} \|E_k - E\|_F = 0) = 1.
					\end{aligned}
				\end{equation*}
				Since $\lim_{k\rightarrow\infty}D_k E_k = DE$ when $\lim_{k\rightarrow\infty}D_k = D$ and $\lim_{k\rightarrow\infty}E_k = E$ are satisfied, we have 
					\begin{equation*}
					\begin{aligned}
						&\mathbb{P}(\lim_{k\rightarrow \infty} \|D_kE_k - DE\|_F = 0) \\
						\ge &\mathbb{P}(\lim_{k\rightarrow \infty} \|D_k - D\|_F = 0\text{\,and\,} \lim_{k\rightarrow \infty} \|E_k - E\|_F = 0)\\
						 = &1,
					\end{aligned}
				\end{equation*}
				where the  last equation uses the fact that if $\mathbb{P}(A)=1$ and $\mathbb{P}(B)=1$, then $\mathbb{P}(A\bigcap B) = 1$.
				\item It is worth noting that $D_k$ might be singular for some  $k\in \mathbb{Z}^+$. However, since $D_k$ converges to the nonsingular matrix $D$ almost surely, there exists sufficiently large $k_0$ such that $D_k$ is nonsingular  for all $k \ge k_0$ with probability 1. Therefore $D_k^\dagger = D_k^{-1}$, $\forall k>k_0$, with probability 1.  Since the inverse of a  nonsingular matrix is a continuous function of the elements of the matrix\cite{stewart1969continuity}, by Lemma~\ref{lemma:AlmostSureContinuous}, we get $D_k^{\dagger}\stackrel{a.s.}{\longrightarrow}D^{-1}$.			
			\end{enumerate}
		\end{proof}
	
				\begin{lemma}
			\label{lemma:WishartDistribution}
			\textnormal{(\cite[Definition 3.4.1]{mardia1979multivariate})} Let $M = \sum_{k=1}^{m}x_kx_k^T$, where  $x_i\sim {N}(0,\Sigma)\in \mathbb{R}^{p } $. Then $M$ has a Wishart distribution with scale matrix $\Sigma $ and degrees of freedom parameter $m$, and we denote it as $M\sim W_p(\Sigma,m)$.
			
			The following properties of Wishart Distribution is helpful:
			\begin{enumerate}[1)]
				\item \textnormal{(\cite[Definition 3.4.1]{mardia1979multivariate})} $\mathbb{E}(M) = m\Sigma$.
				\item \textnormal{(\cite[Theorem 3.1.4]{muirhead2009aspects})} The matrix $M$ is positive definite with probability 1 if and only if  $m \ge p$.
				\item \textnormal{(\cite[Chapter 3.8]{mardia1979multivariate})} If $M\sim W_p(\Sigma,m)$ where $m\ge p$, then $M^{-1}$ is said to have an Inverse-Wishart distribution $W_p^{-1}(\Sigma,m)$. In addition, if $m \ge p+2$, $\mathbb{E}(M^{-1})=\frac{\Sigma^{-1}}{m-p-1}$.		
				\item 	\textnormal{(\cite[Chapter 9.3.1]{holgerssonrecent})}  Let $M \sim W_p(\Sigma,m), m>p$. Then\\
					$\mathbb{E}[M^2] = (m^2+m)\Sigma^2 + m \operatorname{Tr}(\Sigma)\Sigma$, \\
					$\mathbb{E}[{M}^{-2}]=\frac{(m-p-1)}{d} \Sigma^{-2}+\frac{1}{d} \operatorname{Tr}[\Sigma^{-1}] \Sigma^{-1}$ for $m-p-3>0$, where $d = (m-p)(m-p-1)(m-p-3)$.
			\end{enumerate}
		\end{lemma}

		\begin{proof}[\textbf{Proof of Lemma~\ref{lemma:stochastic method error}}]
			It is worth noting that $\tilde{X}_i  $ might be singular, thus we need to make some preparations before applying Lemma~\ref{lemma: a.s. mat prod mat inv} and Lemma~\ref{lemma:WishartDistribution}.
			
			Note that  $\tilde{X}_i  = \tilde{X}_i^T \ge 0$, it has the following decomposition
			\begin{equation}
				\label{eq:proof of lemma 3 eq 1}
				\tilde{X}_i = V_i^T  \begin{bmatrix}
					\Sigma_{i} &\\ &\boldsymbol{O}_{n-r_i}
				\end{bmatrix} V_i,
			\end{equation} where $V_i$ is an orthogonal matrix and  ${\Sigma}_i > 0$ is a diagonal matrix. Together with \eqref{eq:stochastic method eq1}, we get
			\begin{equation*}
				V_i \tilde{X}_i V_i^T  = \sum_{j=1}^N (l_{ij}^{(\gamma)})^2 V_i Y_j^T Y_j V_i^T = \begin{bmatrix}
					\Sigma_{i} &\\ &\boldsymbol{O}_{n-r_i}
				\end{bmatrix},
			\end{equation*} therefore $V_i Y_j^T$ must be in the form
			\begin{equation*}
				V_i Y_j^T = \begin{bmatrix}
					*\\\boldsymbol{O}_{n-r_i}
				\end{bmatrix}, \quad j \in \mathcal{J}_{i}^{(\gamma)}.
			\end{equation*}
			
			By \eqref{eq:stochastic method eq1}, $\tilde{\Upsilon}_{i,k} = \frac{1}{k}\sum_{s=1}^{k} \sum_{r,t=1}^{N}l_{ir}^{(\gamma)} l_{it}^{(\gamma)} Y_r^T\theta_{r,s}\theta_{t,s}^T Y_t $, and 
			therefore $V_i \tilde{\Upsilon}_{i,k} V_i^T$ must be in the form
			\begin{equation}
					\label{eq:proof of lemma 3 eq 2}
				V_i \tilde{\Upsilon}_{i,k} V_i^T= \begin{bmatrix}
					\Sigma_{i,k} &\\ &\boldsymbol{O}_{n-r_i}
				\end{bmatrix},
			\end{equation}
			where $\Sigma_{i,k}\ge 0$.
			
			 Denote $E_i = \begin{bmatrix} I_{r_i}\\ \boldsymbol{O}_{n-r_i}\end{bmatrix}$, \eqref{eq:proof of lemma 3 eq 1} and \eqref{eq:proof of lemma 3 eq 2} can be summarized as
			\begin{equation}
				\label{eq: proof of lemma3 eq5}
					\tilde{X}_i = V_i^T  E_i  \Sigma_i  E_i^T V_i, \;\tilde{\Upsilon}_{i,k} =   V_i^T  E_i  \Sigma_{i,k} E_i^T V_i.
			\end{equation}

			From $\theta_{i,k}\sim \mathbf{N}(0, I_n)$ and \eqref{eq:stochastic method eq2}, \eqref{eq:stochastic method eq1}, we have
			$ \tilde{\upsilon}_{i,k} \sim \mathbf{N}_d(0,  \tilde{X}_i)$ and thus $E_i^TV_i \tilde{\upsilon}_{i,k} \sim \mathbf{N}(0, \Sigma_{i})$. Together with  Lemma~\ref{lemma:WishartDistribution} and 
			\begin{equation*}
				\begin{aligned}
					\Sigma_{i,k} &= E_i ^T V_i \tilde{\Upsilon}_{i,k} V_i^T E_i \\
					&= \frac{1}{k}\sum_{s=1}^{k} E_i ^T V_i\tilde{\upsilon}_{i,s}^{(\gamma)} (E_i ^T V_i\tilde{\upsilon}_{i,s}^{(\gamma)})^T, 
				\end{aligned}
			\end{equation*}
			we can conclude that $\Sigma_{i,k} \sim W_{r_i}(\frac{1}{k}\Sigma_i,k)$. Since $k >n+3 > r_i$, it still holds that
			$\mathbb{P}(\Sigma_{i,k} >0 ) = 1$ and $\Sigma_{i,k}^{-1} \sim W_{r_i}^{-1}(\frac{1}{k}\Sigma_i,k)$.
			
			Now we are ready to verify the assertions in this lemma.
			
			\begin{enumerate}[1)]
			
			\item	By Lemma~\ref{lemma:SLLN for matrix}, it holds that $ \tilde{\Upsilon}_{i,k} \stackrel{a.s.}{\longrightarrow} \tilde{X}_i $, and hence $\Sigma_{i,k}  \stackrel{a.s.}{\longrightarrow} \Sigma_i$ using \eqref{eq: proof of lemma3 eq5}. By Lemma~\ref{lemma: a.s. mat prod mat inv}, $\Sigma_{i,k}^{-1}  \stackrel{a.s.}{\longrightarrow} \Sigma_i^{-1}$ and thus $ 	\tilde{\Upsilon}_{i,k}^\dagger \stackrel{a.s.}{\longrightarrow}   \tilde{X}_i^\dagger $.
			
			\item 
			Considering \eqref{eq:stochastic method eq1} and Lemma~\ref{lemma:WishartDistribution}, it can be verified directly.
			
			\item By exploiting Lemma~\ref{lemma:WishartDistribution} and \eqref{eq: proof of lemma3 eq5}, we have
			\begin{equation*}
				\begin{aligned}
					&\textstyle \mathbb{E}[\tilde{\Upsilon}_{i,k}] = V_i^T E_i \mathbb{E}[\Sigma_{i,k}] E_i^TV_i= \tilde{X}_i,\\
					&\textstyle \mathbb{E}[\tilde{\Upsilon}_{i,k}^{\dagger}] = V_i^T E_i \mathbb{E}[\Sigma_{i,k}^{-1}] E_i^TV_i = \frac{k}{k -r_i -1}\tilde{X}_i^{\dagger}. 
				\end{aligned}
			\end{equation*}
			
			\item By exploiting Lemma~\ref{lemma:WishartDistribution}, denote $k_1:= k-r_i$,  we have that
		\begin{equation}
				\begin{aligned}
					\label{eq: proof of lemma 3 eq 3}
					&\mathbb{E}[\|\Sigma_{i,k} - \Sigma_i\|_F^2 ]\\
					=  &\mathbb{E}[\operatorname{Tr}[(\Sigma_{i,k} - \Sigma_i)^T(\Sigma_{i,k} - \Sigma_i)]] \\
					=  &\operatorname{Tr}(\mathbb{E}[\Sigma_{i,k}^2 ]) -  2 \operatorname{Tr}(\mathbb{E}[\Sigma_{i,k}] \Sigma_i) + \operatorname{Tr}(\Sigma_i^2 ) \\
					= &\textstyle \frac{1}{k} \operatorname{Tr}( \Sigma_i^2) + \frac{1}{k}[ \operatorname{Tr}( \Sigma_i)]^2,\\
					&\mathbb{E}[\|  \Sigma_{i,k}^{-1} - \Sigma_i^{-1} \|_F^2 ]	\\
					=  &\mathbb{E}[\operatorname{Tr}[(\Sigma_{i,k}^{-1} - \Sigma_i^{-1})^T(\Sigma_{i,k}^{-1} - \Sigma_i^{-1})]] \\
					= & \operatorname{Tr}(\mathbb{E}[\Sigma_{i,k}^{-2} ]) -  2\operatorname{Tr}(\mathbb{E}[\Sigma_{i,k}^{-1}] \Sigma_i^{-1})  + \operatorname{Tr}(\Sigma_i^{-2} ) \\
					=&\textstyle \frac{k^2 + k (r_i^2 + 2 r_i + 3 ) - (r_i^3 + 4 r_i^2 + 3 r_i)}{(k - r_i - 3) (k - r_i - 1) (k - r_i) } \operatorname{Tr}( \Sigma_i^{-2})\\
					&+ \textstyle \frac{k^2}{(k - r_i - 3) (k - r_i - 1) (k - r_i) } [\operatorname{Tr}(\Sigma_i^{-1})]^2.
				\end{aligned}
			\end{equation}

			Since the Frobenius norm is invariant under multiplication by an orthogonal matrix, we have
			\begin{equation}
				\begin{aligned}
					\mathbb{E}[\|\tilde{\Upsilon}_{i,k}- \tilde{X}_i\|_F^2 ] =& \mathbb{E}[\|V_i (\tilde{\Upsilon}_{i,k}- \tilde{X}_i)V_i^T\|_F^2 ]\\
					=&\mathbb{E}[\|E_i(\Sigma_{i,k} - \Sigma_i)E_i^T\|_F^2 ]\\
					=&\mathbb{E}[\|\Sigma_{i,k} - \Sigma_i\|_F^2 ],\\
				\mathbb{E}[\|  \tilde{\Upsilon}_{i,k}^{\dagger} -\tilde{X}_i^{\dagger} \|_F^2 ] =& \mathbb{E}[\|V_i (\tilde{\Upsilon}_{i,k}^\dagger-\tilde{X}_i^{\dagger})V_i^T\|_F^2 ]\\
					=&\mathbb{E}[\|  E_i(\Sigma_{i,k}^{-1} - \Sigma_i^{-1})E_i^T \|_F^2 ]\\
					=&\mathbb{E}[\|  \Sigma_{i,k}^{-1} - \Sigma_i^{-1} \|_F^2 ].
				\end{aligned}
			\end{equation}
			
			On the other hand, since $\operatorname{Tr}(DE) = \operatorname{Tr}(ED)$ for any matrices $D,E$ of appropriate sizes, it follows from \eqref{eq:proof of lemma 3 eq 1} that
			\begin{equation}
				\label{eq: proof of lemma 3 eq 4}
				\begin{aligned}
				&\operatorname{Tr}( \tilde{X}_i) =  \operatorname{Tr}( \Sigma_i ), &
				&\operatorname{Tr}( \tilde{X}_i^2)  =  \operatorname{Tr}( \Sigma_i^2 ),\\
				&\operatorname{Tr}(\tilde{X}_i^{\dagger}) =  \operatorname{Tr}( \Sigma_i^{-1}), &
				&\operatorname{Tr}( (\tilde{X}_i^{\dagger})^2)=  \operatorname{Tr}( \Sigma_i^{-2} ).
				\end{aligned}
			\end{equation}
			By Combining \eqref{eq: proof of lemma 3 eq 3}-\eqref{eq: proof of lemma 3 eq 4}, the proof is completed.
		\end{enumerate}
		\end{proof}	
		\subsection{Proof of Lemma~\ref{lemma: R dagger}}
		\begin{proof}[\textbf{Proof of Lemma~\ref{lemma: R dagger}}]
				In Algorithms~1-4,  denote 
			\begin{equation*}
				\bar{C} =\col_{j=1}^N(R_j^{-\frac{1}{2}}C_j), \breve{C}_i = \col_{j=1}^N(l_{ij}^{(\gamma)}R_j^{-\frac{1}{2}}C_j),
			\end{equation*}
			we have
			\begin{equation*}
				\tilde{C}_i = \bar{C}^T\breve{C}_i = \breve{C}_i^T \bar{C},\quad \tilde{R}_i = \breve{C}_i^T\breve{C}_i.
			\end{equation*}
			Thus
			\begin{equation*}
				\tilde{C}_i^T(\tilde{R}_i)^\dagger \tilde{C}_i = \bar{C}^T\breve{C}_i (\breve{C}_i^T\breve{C}_i)^\dagger  \breve{C}_i^T \bar{C}.
			\end{equation*}
			
			Denote the number of rows of  $\breve{C}_i$ as $\tilde{m}$ and $\text{rank}(\breve{C}_i ) $ as $ r_i$, we have the singular value decomposition 
			\begin{equation*}
				\breve{C}_i = U_i \begin{bmatrix}
					\tilde{\Sigma}_i & \boldsymbol{O}_{n-r_i}\\ \boldsymbol{O}_{\tilde{m}-r_i} & \boldsymbol{O}_{\tilde{m}-r_i \times n-r_i}
				\end{bmatrix}V_i^T := U_i \Sigma_i V_i^T ,
			\end{equation*}
			where $U_i$ and $V_i $ are  orthogonal matrices, and $\tilde{\Sigma}_i>0$ is a diagonal matrix. Utilizing this decomposition, we get
			\begin{equation}
				\label{eq:lemma R dagger eq1}
				\begin{aligned}
					\tilde{C}_i^T(\tilde{R}_i)^\dagger \tilde{C}_i 
					& =\bar{C}^TU_i \Sigma_i V_i^T ( V_i \Sigma_i^T \Sigma_i V_i^T)^\dagger   V_i \Sigma_i^T U_i^T \bar{C}\\
					& = \bar{C}^TU_i \Sigma_i ( \Sigma_i^T \Sigma_i )^\dagger \Sigma_i^T U_i^T \bar{C}\\
					& =  \bar{C}^TU_i \begin{bmatrix}
						I_{r_i} &\\& \boldsymbol{O}_{\tilde{m}-r_i}
					\end{bmatrix} U_i^T \bar{C}.
				\end{aligned}
			\end{equation}
			On the other hand, let 
			\begin{equation*}
				\breve{R}_i = V_i\begin{bmatrix}
					\tilde{\Sigma}_i^2 &\\& \boldsymbol{I}_{n-r_i}
				\end{bmatrix} V_i^T.
			\end{equation*}
			It is obvious that $\breve{R}_i$ is positive definite and satisfies
			\begin{equation}
				\label{eq:lemma R dagger eq2}
				\begin{aligned}
					&\tilde{C}_i^T(\breve{R}_i)^{-1} \tilde{C}_i \\
					= &\bar{C}^TU_i \Sigma_i V_i^T 
					(V_i\begin{bmatrix}
						\tilde{\Sigma}_i^2 &\\& \boldsymbol{I}_{n-r_i},
					\end{bmatrix} V_i^T)^{-1}   
					V_i \Sigma_i^T U_i^T \bar{C}\\
					=&  \bar{C}^TU_i \begin{bmatrix}
						I_{r_i} &\\& \boldsymbol{O}_{\tilde{m}-r_i}
					\end{bmatrix} U_i^T \bar{C}.
				\end{aligned}
			\end{equation}
			
			By comparing \eqref{eq:lemma R dagger eq1} and \eqref{eq:lemma R dagger eq2}, the proof is completed.
			\end{proof}
	
	\bibliographystyle{IEEEtran}
	\bibliography{F013W02V3}
	
	\begin{IEEEbiography}[{\includegraphics[width=1in,height=1.25in,clip,keepaspectratio]{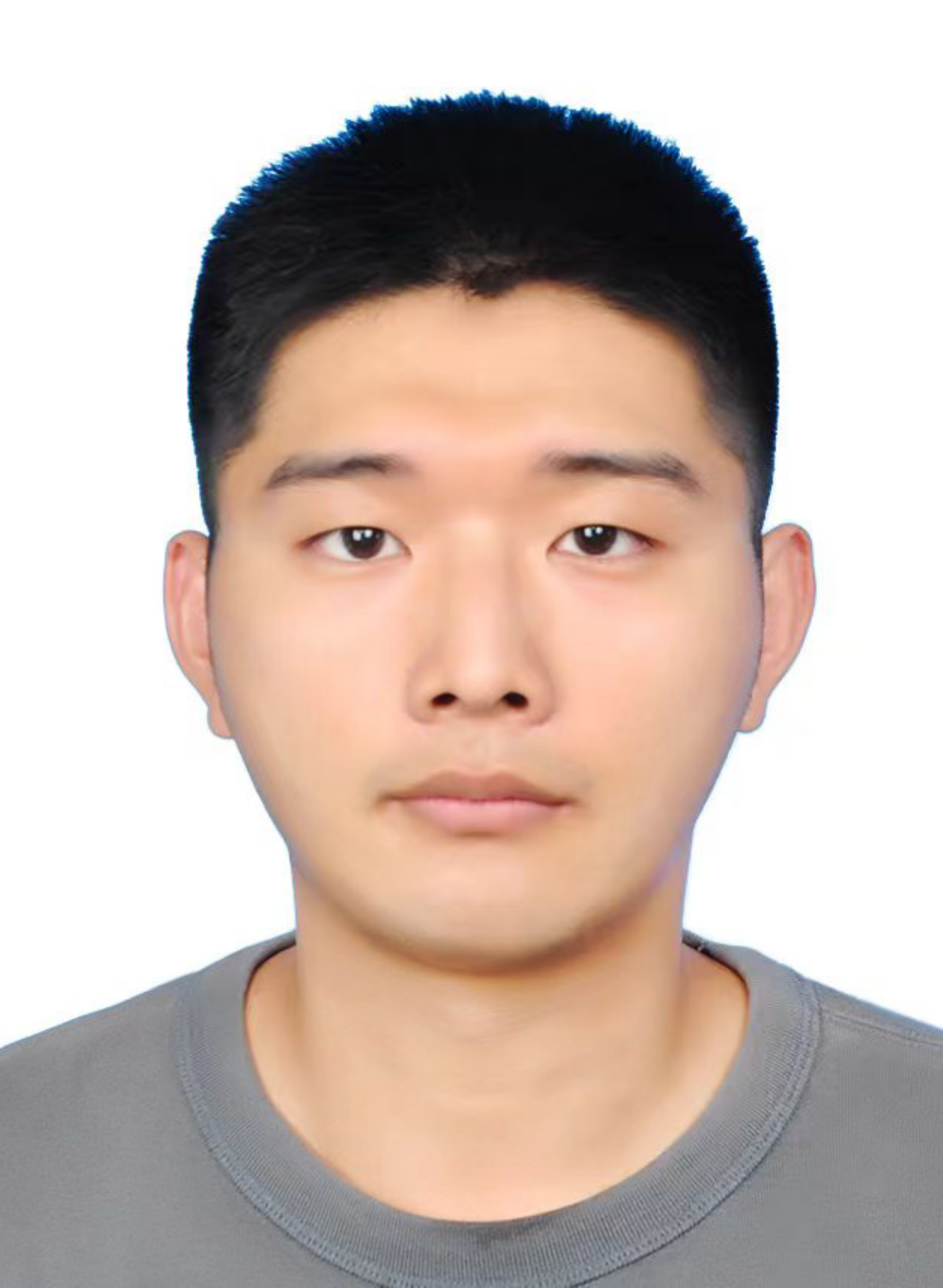}}]{Tuo Yang}
		received the B.E. degree from the College of Engineering,
		Peking University, Beijing, China, in 2021.
		Currently, he is working towards the Ph.D. degree
		at College of Engineering, Peking University.
		His research interests include distributed state estimation, moving-horizon estimation and cooperative control of multi-agent systems.
	\end{IEEEbiography}
	
	\begin{IEEEbiography}
		[{\includegraphics[width=1in,height=1.25in,clip,keepaspectratio]{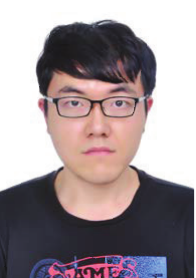}}]{Jiachen Qian}
		\textcolor{black}{received his B.S. degree in College of
			Engineering from Peking University, Beijing, China,
			in 2019. He is now a Ph.D. candidate in Department of Mechanics and Engineering Science, College
			of Engineering at Peking University. His research
			interests include distributed state estimation, event-based state estimation and cooperative control of
			networked systems.}
	\end{IEEEbiography}

	\begin{IEEEbiography}[{\includegraphics[width=1in,height=1.25in,clip,keepaspectratio]{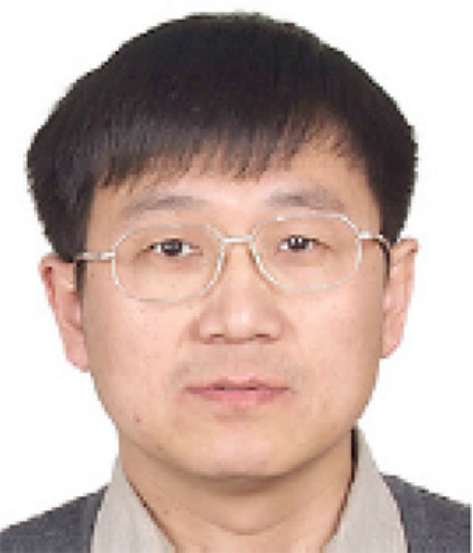}}]{Zhisheng Duan} (Senior Member, IEEE)
		received the M.S. degree in Mathematics from Inner Mongolia University, China and the Ph.D. degree in Control Theory from Peking University, China in 1997 and 2000, respectively.
		From 2000 to 2002, he worked as a Postdoctor in Peking University. Since 2008, he has been a full Professor with the Department of Mechanics and Engineering Science, College of Engineering, Peking University. His research interests include
		robust control, stability of interconnected systems, flight control, and analysis and control of complex dynamical networks. 
		
		Professor Duan received the Guan-Zhao Zhi Best Paper Award at the 2001 Chinese Control Conference and the 2011 first class Award in Natural Science from Chinese Ministry of Education. He obtained the outstanding young scholar award from the National Natural Science Foundation in China and he is currently a Cheung Kong Scholar in Peking University.
	\end{IEEEbiography}

	\begin{IEEEbiography}[{\includegraphics[width=1in,height=1.25in,clip,keepaspectratio]{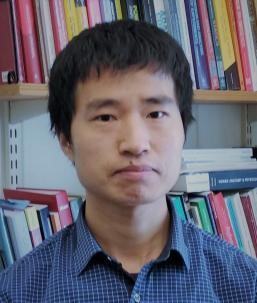}}]{Zhiyong Sun} (Member, IEEE) received the Ph.D. degree from The Australian National University (ANU) Canberra ACT, Australia, in February 2017. He was a Research Fellow/Lecturer with the Research School of Engineering, ANU, from 2017 to 2018. From June 2018 to January 2020, he worked as a postdoctoral researcher at the Department of Automatic Control, Lund University of Sweden. Since January 2020 he has joined Eindhoven University of Technology (TU/e) as an assistant professor. He has won the Springer Best PhD Thesis Award, and several best paper and student paper awards from IEEE CDC, AuCC, ICRA and CCTA. His research interests include multi-agent systems, control of autonomous formations, distributed control and optimization.  He has authored the book titled Cooperative Coordination and Formation Control for Multi-agent Systems (Springer, 2018).
	\end{IEEEbiography}

\end{document}